\documentclass[11pt]{article}
\usepackage[utf8]{inputenc}
\usepackage[a4paper,top=18mm,bottom=20mm,left=20mm,right=20mm]{geometry}
\usepackage{amsmath,amssymb,amsfonts,amsthm,mathtools}
\usepackage[nocompress]{cite}
\usepackage{microtype}
\usepackage{array,booktabs}
\usepackage{graphicx}
\usepackage{tikz}
\usetikzlibrary{arrows.meta}
\usepackage[hidelinks]{hyperref}
\usepackage[capitalise]{cleveref}

\newtheorem{theorem}{Theorem}[section]
\newtheorem{corollary}[theorem]{Corollary}
\newtheorem{lemma}[theorem]{Lemma}
\newtheorem{proposition}{Proposition}[section]

\theoremstyle{definition}
\newtheorem{definition}{Definition}[section]

\theoremstyle{remark}

\newcommand{\D}{\Delta}
\newcommand{\RR}{\mathbb{R}}
\newcommand{\cC}{\mathcal{C}}
\newcommand{\cF}{\mathcal{F}}
\newcommand{\cK}{\mathcal{K}}
\newcommand{\cL}{\mathcal{L}}
\newcommand{\cR}{\mathcal{R}}
\newcommand{\cS}{\mathcal{S}}
\newcommand{\non}{\nonumber}

\numberwithin{equation}{section}
\allowdisplaybreaks[2]
\begin{document}

\thispagestyle{empty}
\begin{center}
	{\LARGE\bfseries Superintegrable systems in curved spaces via coalgebras and deformed Racah algebras}
\end{center}
\vspace{0.3cm}

\begin{center}
	\textsc{Samuel Jimenez-Swain$^{1,\star}$, Ian Marquette$^{2,\dagger}$, and Yao-Zhong Zhang$^{3,\ddagger}$}
\end{center}

\begin{center}
    $^1$ Department of Mathematics, Imperial College London, London SW7 2AZ, UK
\\
	$^2$ Department of Mathematical and Physical Sciences, La Trobe University, Bendigo, VIC 3552, Australia
\\
	$^3$ School of Mathematics and Physics, The University of Queensland, Brisbane, QLD 4072, Australia
\end{center}
\begin{center}
	{\footnotesize $^\star$\textsf{samuel.jimenez-swain26@imperial.ac.uk}\hspace{0.25cm}$^\dagger$\textsf{i.marquette@latrobe.edu.au}\hspace{0.25cm}$^\ddagger$\textsf{yzz@maths.uq.edu.au}}
\end{center}
\vspace{1cm}
\hrule
\begin{abstract}
\noindent We determine the Poisson algebras generated by partial Casimirs of undeformed and non-standard $q$-deformed $\mathfrak{sl}(2)$ coalgebras, respectively.  The derivation is performed on the abstract algebraic level and does not depend on any particular symplectic realization.  For the 3- and 4-site cases the constructions associated with the undeformed $\mathfrak{sl}(2)$ produce the rank-one and rank-two Racah Poisson algebras, respectively, while those with the non-standard deformed $\mathfrak{sl}$  yield finite quadratic Poisson algebras generated by connected partial Casimirs.  At both the three and four sites, the final closure coefficients do not explicitly depend on the deformation parameter.  The results hold at the  Poisson coalgebra level and hence apply to any realizations which preserve the defining brackets and coproducts.  Canonical symplectic realizations give rise to Hamiltonians in spaces of non-constant curvature.  For the St\"ackel $q$-deformed Smorodinsky-Winternitz model we show explicitly how the universal coalgebra sector is completed to maximal superintegrability.  An exact two-block identity, proved directly in the localized Poisson algebra without appealing to a symplectic realization, controls the additional St\"ackel integrals.  In three dimensions, it generates a quadratic Poisson block with its own Casimir relation; in four dimensions, its coupling to the full left/right coalgebra family is organized into three nested quadratic blocks, and the same identity extends this nested coupling to arbitrary $N$.  Explicit higher-site partial Casimirs and factorized bracket generators are collected in the appendices.
\end{abstract}
\vspace{0.35cm}
\hrule

\section{Introduction}

Superintegrability concerns Hamiltonian systems that possess more independent integrals of motion than are required for Liouville integrability.  In the maximally superintegrable case an $n$-dimensional Hamiltonian admits $2n-1$ functionally independent integrals.  Such systems are highly constrained and, in many familiar settings, exhibit multiseparability, exact solvability and closed bounded trajectories.  The systematic study of superintegrability goes back to the classification of two-dimensional Euclidean systems in the 1960s and was subsequently extended to spaces of constant and non-constant curvature; see, for example, \cite{post13,win65,evan90}.  These developments also revealed strong links with special functions, orthogonal polynomials and representation theory.

A complementary line of work concerns the algebras generated by the integrals themselves.  For second-order systems, this often produces quadratic Poisson or associative symmetry algebras \cite{gran92,das01,das06,kal03}.  Higher-dimensional and block-separated examples also often exhibit quadratic symmetry algebras \cite{tan10,kmp10,kmp11,faz18,jun23}, while genuinely higher-order integrals lead to higher-order polynomial closures \cite{ian2010,parr2025}.  The Racah algebra is one of the central examples.  Its role in generic superintegrability and recoupling is reviewed and developed in \cite{post13,gvz14,vin19,bie21}.  Higher-rank polynomial algebras and related models are described in \cite{dbgv18,cor21}, centralizer and commutant approaches in \cite{cram22,jun24}, and chain and embedding constructions in \cite{liao,lat,lat21,lat21b}.  The appearance of such polynomial symmetry algebras is not merely a by-product of a particular Hamiltonian. They are algebraic structures whose Casimirs, substructures and representations encode spectral and separation information of the underlying systems.

The coalgebra formalism provides a direct mechanism for constructing integrable and superintegrable families.  Starting from a Poisson algebra equipped with a coassociative coproduct and a Casimir function, iterated coproducts produce left and right chains of partial Casimirs.  Early applications include superintegrable $q$-deformations of the Smorodinsky-Winternitz family \cite{bal04}.  More generally, for a large class of Hamiltonians built from total coproduct generators, the partial Casimirs give the universal $2n-3$ coalgebra integrals and, together with the Hamiltonian, the familiar quasi-maximal superintegrability is achieved \cite{bal05,bal09b}.  An important feature is the flexibility of the Hamiltonian: the coalgebra symmetry survives for arbitrary smooth functions of the total coproduct generators.  In the $\mathfrak{sl}(2)$ setting, this construction includes deformations associated with spaces of non-constant curvature \cite{bal07,bal08}.

Most explicit coalgebra calculations in the superintegrability literature are performed after choosing a symplectic realization.  This is appropriate from the dynamical point of view, but it can obscure which identities depend on the realization and which are already consequences of the abstract Poisson coalgebra.  One aim of the present paper is to carry out the closure problem directly on tensor products of the underlying Poisson algebra.  The elementary variables are copies of the abstract generators, the Poisson bracket is the corresponding product Poisson structure, and the partial Casimirs are computed before any realization is imposed.  Once the closure is established at this level, every Poisson realization inherits the same relations.  We first rederive the undeformed $\mathfrak{sl}(2)$ case without choosing a realization.  For three sites the familiar rank-one Racah Poisson algebra is recovered.  At four sites the same method reproduces the rank-two Racah-type structure and makes transparent the role of the different two- and three-site partial Casimirs.  This calculation is useful as a reference because the primitive coproduct permits non-contiguous objects such as $C_{13}$ to be introduced in a simple symmetric way.

The main results of the paper are those associated with non-standard Poisson deformation $\mathfrak{sl}_z(2)$.  The underlying non-standard Poisson coalgebra and its curved-space realizations are known \cite{bal05,bal09b,bal07}.  Here we derive the explicit $N=3$ and $N=4$ partial Casimir closure directly in the abstract product Poisson algebra, before fixing a realization.  In this setting both the Poisson brackets and the coproduct are deformed, and the latter contains exponential factors.  Consequently, the partial Casimirs are no longer polynomial functions of the elementary generators and the undeformed choice of a minimal-site Racah basis is no longer available in the same form.

It is not automatic that the Poisson brackets among the coalgebra integrals should close in a finite polynomial algebra.  We show that they do so at the first two nontrivial ranks.  For three sites the algebra is generated by $C_{12}$, $C_{23}$ and their bracket $\mathcal F_{12,23}$, with the one-site Casimirs and $C_{123}$ central.  For four sites the connected partial Casimirs $C_{12}$, $C_{23}$, $C_{34}$, $C_{123}$ and $C_{234}$ require four bracket generators.  The complete closure remains quadratic, and its displayed structure coefficients are again independent of $z$ even though the generators themselves depend non-polynomially on the deformation.

The non-standard deformation considered here should be kept distinct from the standard Drinfeld--Jimbo Poisson-Hopf and quantum deformation.  Both deform the same undeformed coalgebraic mechanism and both recover Racah-type structures as the deformation is removed, but their closures at finite deformation are organized differently.  In the present Jordanian-type Poisson problem the additional generators are Poisson brackets such as $\mathcal F_{12,23}$, whereas in the standard deformation non-contiguous quantities are described through oriented or braided intermediate Casimirs.  The distinction is structural and prevents the two algebras at finite deformation from being conflated.

This result should be distinguished from a statement about one particular curved Hamiltonian.  The calculations depend only on the defining Poisson brackets and the coproduct, and hence apply to every realization of the same Poisson coalgebra.  We return to the canonical symplectic realization in the final part of the paper to display the associated integrals explicitly and to recall the geodesic, oscillator and Kepler--Coulomb Hamiltonians on spaces of non-constant curvature.  The formulae corresponding to symplectic realizations are also supplied to provide an independent way of checking the abstract closure.

The paper is organized as follows.  Section~2 reviews integrability, Poisson coalgebras, left and right partial Casimirs, $q$-deformed Poisson coalgebras and the closure problem.  Section~3 develops the undeformed $N=3$ and $N=4$ Racah constructions directly in the tensor-product Poisson algebra.  Section~4 treats the non-standard deformation at three sites, and Section~5 gives the full four-site algebra and its quadratic closure.  Section~6 discusses Hamiltonians and the canonical symplectic realization, then specializes to the maximally superintegrable St\"ackel $q$-Smorodinsky-Winternitz model.  The model-specific algebra is derived from an exact abstract two-block lemma, specialized first in three dimensions and then to the complete four-dimensional coupling of the missing integral with the coalgebra family; a short corollary records the arbitrary-$N$ nested extension.  Appendix~A gives compact explicit higher-site partial Casimirs, while Appendix~B records factorized phase-space formulae for the four bracket generators used in the realized $N=4$ algebra.

\section{Coalgebra and superintegrable Hamiltonian systems}

The purpose of this section is to fix the algebraic setting used throughout the paper.  We keep the distinction between the abstract Poisson coalgebra and a later canonical realization explicit.  This is important in the deformed problem because the coproduct contains transcendental functions of the generators, while the closure relations remain polynomial in the partial Casimirs.

\subsection{Classical integrability and superintegrability}
A Hamiltonian system with $n$ degrees of freedom is Liouville integrable when it possesses $n$ functionally independent integrals in involution.
\begin{definition}
A classical Hamiltonian system is \emph{integrable} if it admits $n$ functionally independent integrals $X_1,\ldots,X_n$ satisfying
\begin{equation}
 \{X_i,X_j\}=0,\qquad i,j=1,\ldots,n.
\end{equation}
It is \emph{superintegrable} if it admits $n+k$ functionally independent integrals with $1\leq k\leq n-1$.  The cases $k=1$ and $k=n-1$ are referred to as minimal and maximal superintegrability, respectively; for $n\geq3$, the case $k=n-2$ is called quasi-maximal superintegrability.
\end{definition}

The number of integrals is only one aspect of the problem.  Their polynomial degree, their mutual Poisson brackets and the existence of several maximal commuting subsets may carry additional dynamical information.  For instance, second-order maximally superintegrable systems on spaces of constant curvature are typically multiseparable, while higher-order systems may possess polynomial symmetry algebras of higher degree.  Other notions of integrability, including noncommutative or degenerate integrability, are useful in broader geometric settings.  In the present paper, however, the basic object is an algebra of Poisson-commuting or noncommuting integrals generated by a coalgebra construction.

On canonical phase space we use the convention
\begin{definition}
For functions $\cR(\boldsymbol q,\boldsymbol p)$ and $\cS(\boldsymbol q,\boldsymbol p)$,
\begin{equation}
 \{\cR,\cS\}=\sum_{j=1}^{n}\left(
 \frac{\partial\cR}{\partial q_j}\frac{\partial\cS}{\partial p_j}
 -\frac{\partial\cR}{\partial p_j}\frac{\partial\cS}{\partial q_j}
 \right).
 \label{canonicalPBcorrected}
\end{equation}
Hence $\{q_j,p_k\}=\delta_{jk}$ and $\{q_j,q_k\}=\{p_j,p_k\}=0$.
\end{definition}
The sign convention in \eqref{canonicalPBcorrected} is the one used by the symplectic realization in Section~6.

The canonical bracket is only one realization of the abstract structures considered below.  A Lie-Poisson algebra may be specified by
\begin{equation}
 \{X_i,X_j\}=C_{ij}^{k}X_k,
\end{equation}
while a deformed Poisson algebra has nonlinear brackets
\begin{equation}
 \{X_i,X_j\}=f_{ij}(X_1,\ldots,X_M).
\end{equation}
A realization is a Poisson map sending the abstract generators to functions of canonical variables.  The closure relations derived before choosing such a map are consequently common to all realizations that preserve the defining bracket.

\begin{definition}
The Poisson algebra generated by the integrals of motion is called the \emph{symmetry algebra}.  A polynomial closure has the schematic form
\begin{equation}
 \{X_i,X_j\}
 =\sum_k a_{ij}^{k}X_k
 +\sum_{k,\ell}b_{ij}^{k\ell}X_kX_\ell+\cdots.
 \label{eq:polyclosure}
\end{equation}
\end{definition}
The original list of integrals can be enlarged by new generators defined as Poisson brackets.  This does not necessarily increase the number of functionally independent constants of motion; rather, it can turn a complicated nonlinear closure into a finite polynomial presentation.  We use this strategy for the non-standard deformation.

\subsection{Poisson coalgebras and partial Casimirs}
Coalgebra methods encode how one copy of an algebraic system is embedded into several tensor factors.  Let $A$ be an unital associative algebra equipped with a coproduct $\Delta:A\to A\otimes A$.  We require coassociativity,
\begin{equation}
 (\Delta\otimes\mathrm{id})\circ\Delta
 =(\mathrm{id}\otimes\Delta)\circ\Delta,
 \label{eq:originalcoprod}
\end{equation}
and the algebra-morphism property $\Delta(ab)=\Delta(a)\Delta(b)$. Strictly speaking, this is part of the structure of a bialgebra; following the superintegrability literature, we use the term coalgebra because the coproduct and its Casimir construction are the structures needed here.

If $A$ is a Poisson algebra, the coproduct must also preserve the Poisson bracket,
\begin{equation}
 \Delta(\{a,b\}_A)=\{\Delta(a),\Delta(b)\}_{A\otimes A}.
 \label{eq:poissonmorphism}
\end{equation}
On elementary tensors the product Poisson bracket is
\begin{equation}
 \{a\otimes b,c\otimes d\}_{A\otimes A}
 =\{a,c\}_{A}\otimes bd+ac\otimes\{b,d\}_{A}.
 \label{eq:fourtensorcom}
\end{equation}
The iterated coproduct is defined recursively by
\begin{equation}
 \Delta^{(m)}=(\Delta^{(m-1)}\otimes\mathrm{id})\circ\Delta,
 \qquad \Delta^{(2)}=\Delta,
 \label{eq:cas-nm}
\end{equation}
and is again both algebra and Poisson morphism.  Coassociativity guarantees that the same $m$-site object is obtained independently of the placement of parentheses.

Let $C_a$, $a=1,\ldots,r$, be Casimir functions of the one-site Poisson algebra, so that $\{C_a,y^\alpha\}=0$ for every generator $y^\alpha$.  Whenever a $m$-site coproduct with $m<N$ is compared with a $N$-site element below, we use the natural left embedding $\Delta^{(m)}(x)\mapsto\Delta^{(m)}(x)\otimes\mathbf 1^{\otimes(N-m)}$ into $A^{\otimes N}$.  Applying the Poisson morphism property gives
\begin{equation}
 \{\Delta^{(m)}(C_a),\Delta^{(N)}(y^\alpha)\}=0,
 \qquad m=2,\ldots,N.
 \label{poissoncommutegens}
\end{equation}
In particular,
\begin{equation}
 \{\Delta^{(m)}(C_a),\Delta^{(n)}(C_b)\}=0
\end{equation}
whenever the two elements belong to the same nested left chain.  This is the fundamental reason why coproducts of Casimirs provide commuting integrals for Hamiltonians constructed from the total coproduct generators.

For an $N$-site system we make the embeddings into the full tensor product explicit.  The left partial Casimirs are
\begin{equation}
 C^{(m)}=\Delta^{(m)}(C)\otimes\mathbf 1^{\otimes(N-m)},
 \qquad m=2,\ldots,N,
 \label{eq:leftcoprod}
\end{equation}
whereas the right partial Casimirs are
\begin{equation}
 C_{(m)}=\mathbf 1^{\otimes(N-m)}\otimes\Delta^{(m)}(C),
 \qquad m=2,\ldots,N.
 \label{eq:rightcoprod}
\end{equation}
The two $N$-site elements coincide with the total Casimir.  For $m<N$ they are supported on different contiguous blocks.  The two chains are displayed in Table~\ref{tab:coalgebrachains}.

\begin{table}[ht]
\centering
\begin{tabular}{c@{\qquad\qquad}c}
\toprule
Left chain & Right chain\\
\midrule
$C^{(2)}$ on $1,2$ & $C_{(2)}$ on $N-1,N$\\
$C^{(3)}$ on $1,2,3$ & $C_{(3)}$ on $N-2,N-1,N$\\
$\vdots$ & $\vdots$\\
$C^{(m)}$ on $1,\ldots,m$ & $C_{(m)}$ on $N-m+1,\ldots,N$\\
$\vdots$ & $\vdots$\\
$C^{(N)}$ & $C_{(N)}=C^{(N)}$\\
\bottomrule
\end{tabular}
\caption{Left and right partial Casimirs generated from a single Casimir function.  In a canonical $N$-particle realization these two chains supply the standard coalgebraic family of $2N-3$ nontrivial conserved quantities once their common total member is counted once.}
\label{tab:coalgebrachains}
\end{table}

The existence of these commuting chains is not the same problem as determining the full non-Abelian symmetry algebra.  In the undeformed $\mathfrak{sl}(2)$ case, one can introduce further partial Casimirs associated with non-contiguous subsets, such as $C_{13}$, and obtain a particularly symmetric presentation of the Racah algebra.  This freedom is closely related to the primitive coproduct.  In a deformation, by contrast, the coproduct itself carries ordering-dependent exponential factors, so a naive ``skip-site'' construction need not have the same coalgebraic meaning.  For this reason, we distinguish carefully below between the commuting left/right chains and additional elements introduced only for the purpose of algebraic closure.

\subsection[q-deformed Poisson coalgebras]{$q$-deformed Poisson coalgebras}
A Poisson coalgebra deformation replaces $(A,\Delta)$ by a one-parameter family $(A_z,\Delta_z)$ such that $\Delta_z$ remains a Poisson morphism and the undeformed structure is recovered as $z\to0$.  Both the Poisson brackets and the coproduct may become nonlinear. In the examples associated with non-standard deformations of $\mathfrak{sl}(2)$, exponentials of $J_-$ appear in the coproduct and hyperbolic functions appear in the brackets and Casimir \cite{bal05,bal07,bal08}.

At the Poisson level these structures may be viewed as Poisson analogues of quantum-group deformations.  The product on the corresponding Poisson--Lie group is encoded algebraically by the coproduct, and compatibility of that product with the Poisson bracket permits the coalgebra construction to survive the deformation.  The limit $z\to0$ recovers the primitive $\mathfrak{sl}(2)$ Lie--Poisson coalgebra.

The undeformed and deformed constructions can be represented schematically as follows.  We display both coassociative paths, in the same spirit as the undeformed coproduct diagram, because this makes clear which structural property survives the deformation.
\begin{center}
\resizebox{0.576\textwidth}{!}{%
\begin{tikzpicture}[>=stealth,
  every node/.style={font=\normalsize}]

  \node (A)   at (0.0,5.0) {$A$};
  \node (AAu) at (5.0,5.0) {$A\otimes A$};
  \node (AAl) at (0.0,1.65) {$A\otimes A$};
  \node (AAA) at (5.0,1.65) {$A\otimes A\otimes A$};

  \node (Az)     at (2.35,3.30) {$A_z$};
  \node (AzAzu)  at (7.85,3.30) {$A_z\otimes A_z$};
  \node (AzAzl)  at (2.35,-0.20) {$A_z\otimes A_z$};
  \node (AzAzAz) at (7.85,-0.20) {$A_z\otimes A_z\otimes A_z$};

  \draw[->] (A) -- (AAu);
  \draw[->] (A) -- (AAl);
  \draw[->] (AAu) -- (AAA);
  \draw[->] (AAl) -- (AAA);
  \node[above] at (2.50,5.0) {$\Delta$};
  \node[left=7pt] at (0.0,3.35) {$\Delta$};
  \node[right=8pt] at (5.0,4.25) {$\Delta\otimes\mathrm{id}$};
  \node[below=6pt] at (3.85,1.65) {$\mathrm{id}\otimes\Delta$};

  \draw[->] (Az) -- (AzAzu);
  \draw[->] (Az) -- (AzAzl);
  \draw[->] (AzAzu) -- (AzAzAz);
  \draw[->] (AzAzl) -- (AzAzAz);
  \node[above=6pt] at (6.25,3.30) {$\Delta_z$};
  \node[left=9pt] at (2.35,2.48) {$\Delta_z$};
  \node[right=8pt] at (7.85,1.60) {$\Delta_z\otimes\mathrm{id}$};
  \node[below=5pt] at (5.10,-0.20) {$\mathrm{id}\otimes\Delta_z$};

  \draw[->,dashed,gray] (A) -- (Az);
  \draw[->,dashed,gray] (AAu) -- (AzAzu);
  \draw[->,dashed,gray] (AAl) -- (AzAzl);
  \draw[->,dashed,gray] (AAA) -- (AzAzAz);
  \node[text=black,above=5pt] at (1.05,4.25) {$\Phi_z$};
  \node[text=black,above=7pt] at (6.35,4.30) {$\Phi_z\otimes\Phi_z$};
  \node[text=black,below=7pt] at (1.10,0.70) {$\Phi_z\otimes\Phi_z$};
  \node[text=black,above=7pt] at (6.45,0.75) {$\Phi_z^{\otimes3}$};
\end{tikzpicture}%
}
\end{center}
Here $\Phi_z:A\dashrightarrow A_z$ is only a symbolic correspondence indicating the passage from the undeformed to the deformed coalgebra; the dashed arrows and their tensor powers are schematic.  No explicit map $A\to A_z$ is required in the calculations, and the diagram is not meant to assert the strict commutativity of a separately defined deformation homomorphism.  The exact structural statements used below are the coassociativity of $\Delta$ and $\Delta_z$, their Poisson-morphism properties, and the recovery of the undeformed coalgebra in the limit $z\to0$.

The deformation affects more than the explicit appearance of the integrals.  Since both the algebraic relations and the coproduct change, the hierarchy of convenient intermediate quantities can also change.  A generator that has a simple expression on a minimal set of sites in the primitive case may no longer be a convenient coalgebra element at finite $z$.  Determining the symmetry algebra requires an explicit closure calculation rather than a direct transcription of the undeformed Racah presentation.

\subsection{Closure under the Poisson bracket}
The left and right chains provide the integrals needed for the coalgebraic superintegrability statement, but they do not determine the non-Abelian algebra generated when quantities from different chains are combined.  We denote a partial Casimir by the set of sites on which it is supported,
\begin{equation}
 C_1,C_{12},C_{123},\ldots,C_{1\cdots N}
\end{equation}
for the left chain and
\begin{equation}
 C_N,C_{N-1,N},C_{N-2,N-1,N},\ldots,C_{1\cdots N}
\end{equation}
for the right chain.  In the undeformed problem we enlarge this notation to all subset Casimirs,
\begin{equation}
 C_i,\quad C_{ij},\quad C_{ijk},\ldots,
 \label{eq:extendedsubsets}
\end{equation}
although only some of these belong to a given commuting chain.  Retaining algebraically dependent elements can be advantageous because it preserves permutation symmetry and often lowers the degree of the displayed relations.

If a Poisson bracket is not expressible polynomially in the current set of generators, we introduce a new element and record its origin explicitly.  At this level we distinguish the raw Poisson bracket from any normalization adopted later, and write
\begin{equation}
 \mathcal B_{ij,lm}:=\{C_{ij},C_{lm}\}.
 \label{eq:Fnotation}
\end{equation}
The normalization of a bracket-generated element is conventional and will be stated explicitly in each construction.  In the undeformed Racah subsection below we use $\cF_{ij,jk}=\frac12\mathcal B_{ij,jk}$, which is the normalization convenient for the standard Racah relations.  In the non-standard $q$-Poisson sections we instead use $\cF_{A,B}=\mathcal B_{A,B}$, i.e., the unscaled Poisson bracket.  No change of normalization is made within an individual algebra.  The same principle applies to brackets involving three-site Casimirs.  The notation is deliberately redundant: $\{C_{ijk},C_{lmn}\}$ and $\{C_{ij},C_{klmn}\}$ may involve the same union of indices but are different algebraic constructions.  Keeping the pair of parent generators visible avoids this ambiguity.

Starting from a chosen set $\mathcal I=\{I_1,\ldots,I_s\}$, the closure problem is to determine whether
\begin{equation}
 \{I_i,I_j\}
 =\sum_a \alpha^{(1)}_{ij;a}I_a
 +\sum_{a,b}\alpha^{(2)}_{ij;ab}I_aI_b+\cdots
 \label{eq:closureschematic}
\end{equation}
for every pair.  If not, independent brackets are adjoined and the procedure is repeated.  Functional independence is not required at every stage: the aim is a finite algebraic presentation of the symmetry algebra, not a minimal list of constants of motion.  This distinction is important already for the ordinary Racah algebra and becomes essential in the deformed case considered below.

\section{Algebraic Racah algebra and closure based on Poisson bracket}

\subsection[sl(2) with N=3 and Racah algebra]{$\mathfrak{sl}(2)$ with $N=3$ and Racah algebra}

We first recover the rank-one Racah algebra directly in the tensor-product Lie--Poisson algebra, without choosing a symplectic realization.  We use the complexified $\mathfrak{su}(2)$ presentation (equivalently $\mathfrak{sl}_2$) with brackets and primitive coproduct
\begin{align}
    \{J_x,J_y\}&=iJ_0,\qquad \{J_y,J_0\}=iJ_x,\qquad \{J_0,J_x\}=iJ_y\\
    \D(1)&=1\otimes 1,\qquad \D(J_i)=J_i^{(1)}+J_i^{(2)},\quad i=x,y,0
\end{align}
where the Casimir function is given by
\begin{align}
    \cC=J_x^2+J_y^2+J_0^2
\end{align}
To simplify calculations, it is convenient to perform a change of basis using the standard raising and lowering operators
\begin{align}
    J_\pm = J_x \pm i J_y,
\end{align}
which leads to
\begin{align}
    \Delta(J_\pm) = J_\pm^{(1)} + J_\pm^{(2)}.
\end{align}
In this basis, the Lie-Poisson brackets take the familiar form
\begin{align}
    \{J_0, J_\pm\} = \pm J_\pm, \qquad \{J_+, J_-\} = 2 J_0,
\end{align}
and the Casimir function can be written as
\begin{align}
    \mathcal{C} = J_- J_+ + J_0^2.
\end{align}

For three copies, let $\{J_i^{(\alpha)}\}$, $\alpha=1,2,3$ and $i=\pm,0$, denote the generators on the individual sites.  The resulting nine-dimensional Poisson structure is
\begin{align}
    \{J_0^{(\alpha)}, J_\pm^{(\beta)}\} = \pm\delta_{\alpha\beta} J_\pm^{(\alpha)}, \qquad \{J_+^{(\alpha)}, J_-^{(\beta)}\} = 2\delta_{\alpha\beta} J_0^{(\alpha)},\qquad \forall \alpha,\beta=1,2,3 \label{9dimn3a}
\end{align}
where $\delta_{\alpha\beta}$ is the Kronecker delta.  We use the ordered coordinate set

\begin{equation}
\tilde{J}=\{J_{-}^{(1)},J_{+}^{(1)},J_{0}^{(1)},J_{-}^{(2)},J_{+}^{(2)},J_{0}^{(2)},J_{-}^{(3)},J_{+}^{(3)},J_{0}^{(3)}\} \label{9dimstr}
\end{equation}
with Poisson brackets
\begin{equation}
\{\tilde{J}_i,\tilde{J}_j\}= \tilde{C}^k_{ij} \tilde{J}_k  \label{9dimn3b}
\end{equation}
where the structure constants are determined by \cref{9dimn3a}.  Since $\cC$ Poisson-commutes with the one-site generators and the coproduct is a Poisson morphism, $\D^{(m)}(\cC)$ Poisson-commutes with any smooth function $F$ of the corresponding coproduct generators:
\begin{align}
    \{\D^{(m)}(\cC),F(\D^{(m)}(J_-),\D^{(m)}(J_+),\D^{(m)}(J_0))\}=0.
\end{align}
This fact will be used repeatedly in the construction of the integrals of motion.

Following the procedure described in \cite{lat21b}, we define the single, double, and triple coproduct Casimirs as
\begin{align}
    C_\alpha &= J_-^{(\alpha)} J_+^{(\alpha)} + (J_0^{(\alpha)})^2,\\
    C_{12}&=\D^{(2)}(\cC)=\D^{(2)}(J_-)\D^{(2)}(J_{+})+\left(\D^{(2)}(J_0)\right)^2\non\\
    &=(J_-^{(1)}+J_{-}^{(2)})(J_+^{(1)}+J_{+}^{(2)})+(J_0^{(1)}+J_{0}^{(2)})^2\non\\
    &=J_-^{(1)}J_+^{(1)}+J_-^{(1)}J_+^{(2)}+J_-^{(2)}J_+^{(1)}+J_-^{(2)}J_+^{(2)}+(J_0^{(1)})^2+2J_0^{(1)}J_0^{(2)}+(J_0^{(2)})^2,\\
    C_{23}&=\D^{(2)}_R(\cC)=\D^{(2)}_R(J_-)\D^{(2)}_R(J_{+})+\left(\D^{(2)}_R(J_0)\right)^2\non\\
    &=(J_-^{(2)}+J_{-}^{(3)})(J_+^{(2)}+J_{+}^{(3)})+(J_0^{(2)}+J_{0}^{(3)})^2\non\\
    &=J_-^{(3)}J_+^{(3)}+J_-^{(3)}J_+^{(2)}+J_-^{(2)}J_+^{(3)}+J_-^{(2)}J_+^{(2)}+(J_0^{(3)})^2+2J_0^{(3)}J_0^{(2)}+(J_0^{(2)})^2,\\
    C_{123}&=\D^{(3)}(\cC)=\D^{(3)}(J_-)\D^{(3)}(J_{+})+\left(\D^{(3)}(J_0)\right)^2\non\\
    &=(J_-^{(1)}+J_-^{(2)}+J_{-}^{(3)})(J_+^{(1)}+J_+^{(2)}+J_{+}^{(3)})+(J_0^{(1)}+J_0^{(2)}+J_{0}^{(3)})^2\non\\
    &= J_-^{(1)}J_+^{(1)} + J_-^{(1)}J_+^{(2)} + J_-^{(1)}J_+^{(3)}
     + J_-^{(2)}J_+^{(1)} + J_-^{(2)}J_+^{(2)} + J_-^{(2)}J_+^{(3)} + J_-^{(3)}J_+^{(1)} + J_-^{(3)}J_+^{(2)}\non\\
    &\quad  + J_-^{(3)}J_+^{(3)}
    + (J_0^{(1)})^2 + (J_0^{(2)})^2 + (J_0^{(3)})^2  + 2J_0^{(1)}J_0^{(2)} + 2J_0^{(1)}J_0^{(3)} + 2J_0^{(2)}J_0^{(3)}.
\end{align}
For the undeformed primitive coproduct it is also useful to introduce the non-contiguous two-site Casimir $C_{13}$,
\begin{align}
    C_{13}&=(J_-^{(1)}+J_{-}^{(3)})(J_+^{(1)}+J_{+}^{(3)})+(J_0^{(1)}+J_{0}^{(3)})^2\non\\
    &=J_-^{(1)}J_+^{(1)}+J_-^{(1)}J_+^{(3)}+J_-^{(3)}J_+^{(1)}+J_-^{(3)}J_+^{(3)}+(J_0^{(1)})^2+2J_0^{(1)}J_0^{(3)}+(J_0^{(3)})^2.
\end{align}

These elements are not independent, as they satisfy the linear relation
\begin{align}
    C_{123} = C_{12} + C_{23} + C_{13} - C_1 - C_2 - C_3. \label{eq:su2-n3-triplecas}
\end{align}
We consider the set
\begin{align}
    \cK_0&:=\{C_{1},C_{2},C_{3},C_{23},C_{13},C_{12},C_{123}\} \label{gen3na}.
\end{align}
All these elements are functions on the nine-generator Poisson space \cref{9dimstr}.  The only nonzero brackets within $\cK_0$ occur among the two-site Casimirs.  In this undeformed Racah subsection we adopt the conventional half-bracket normalization $\cF=\frac12\mathcal B$ introduced after \cref{eq:Fnotation}, and define
\begin{align}
    \cF_{12,23}&:=\frac{1}{2}\{C_{12},C_{23}\}\label{c1223},\\
    \cF_{23,13}&:=\frac{1}{2}\{C_{23},C_{13}\}\label{c1323},\\
    \cF_{13,12}&:=\frac{1}{2}\{C_{13},C_{12}\}\label{c1213}.
\end{align}
Centrality of $C_{123}$ and the linear relation above imply
\begin{align}
    0&=\{C_{123},C_{12}\}=\{C_{13}+C_{23},C_{12}\}=\{C_{13},C_{12}\}+\{C_{23},C_{12}\}=2\cF_{13,12}-2\cF_{12,23},\\
    0&=\{C_{123},C_{13}\}=\{C_{12}+C_{23},C_{13}\}=\{C_{12},C_{13}\}+\{C_{23},C_{13}\}=2\cF_{23,13}-2\cF_{13,12},\\
    0&=\{C_{123},C_{23}\}=\{C_{13}+C_{12},C_{23}\}=\{C_{13},C_{23}\}+\{C_{12},C_{23}\}=2\cF_{12,23}-2\cF_{23,13}.
\end{align}
Hence the three bracket generators coincide:
\begin{align}
    \cF:=\cF_{12,23}=\cF_{23,13}=\cF_{13,12}
\end{align}
and the closure may be written on
\begin{align}
    \cK_1&:=\{C_{1},C_{2},C_{3},C_{23},C_{13},C_{12},C_{123},\cF\}. \label{gen3nb}
\end{align}

Direct evaluation gives
\begin{align}
    \{C_{12},\cF\}&=C_{23}C_{12}-C_{12}C_{13}+(C_2-C_1)(C_3-C_{123}),\label{cf1}\\
    \{C_{23},\cF\}&=C_{13}C_{23}-C_{23}C_{12}+(C_3-C_2)(C_1-C_{123}),\label{cf2}\\
    \{C_{13},\cF\}&=C_{12}C_{13}-C_{13}C_{23}+(C_1-C_3)(C_2-C_{123}).
\end{align}
The constant terms are related by cyclic permutation, and their sum vanishes together with the bilinear terms.
The three relations satisfy
\begin{align}
    \{C_{12},\cF\}+\{C_{23},\cF\}+\{C_{13},\cF\}=0.\label{3sum}
\end{align}
Eliminating $C_{13}$ with \cref{eq:su2-n3-triplecas},
\begin{align}
     C_{13}&=C_{123}-C_{12}-C_{23}+C_{1}+C_{2}+C_{3}
\end{align}
gives
\begin{align}
    \{C_{12},\cF\}&=C_{12}^2+2C_{23}C_{12}-(C_{123}+C_1+C_2+C_3)C_{12}+(C_2-C_1)(C_3-C_{123}),\label{classracah1}\\
    \{C_{23},\cF\}&=-C_{23}^2-2C_{12}C_{23}+(C_{123}+C_1+C_2+C_3)C_{23}+(C_3-C_2)(C_1-C_{123}).\label{classracah2}
\end{align}

One may work either with the symmetric but linearly dependent set containing $C_{13}$, or with a reduced generating set.
\begin{align}
    \cL&:=\{C_{1},C_{2},C_{3},C_{23},C_{12},C_{123},\cF \}. \label{gen3nc}
\end{align}
The reduced set~\cref{gen3nc} uses fewer generators and fewer relations, at the cost of hiding the permutation symmetry that is manifest in the redundant set~\cref{gen3nb}.

Equations~\eqref{classracah1}--\eqref{classracah2} are the $(n,k)=(3,2)$ substructure described in \cite{lat} and give the classical rank-one Racah algebra $R(3)$ \cite{lat21}.

\subsection[sl(2) with N=4 and Racah algebra]{$\mathfrak{sl}(2)$ with $N=4$ and Racah algebra}

For four copies of the same Lie--Poisson algebra, let $\{J_i^{(\alpha)}\}$ denote the generators at sites $\alpha=1,2,3,4$, with $i=\pm,0$.  The corresponding twelve-dimensional Poisson structure is
\begin{align}
    \{J_0^{(\alpha)}, J_\pm^{(\beta)}\} = \pm\delta_{\alpha\beta} J_\pm^{(\alpha)}, \qquad \{J_+^{(\alpha)}, J_-^{(\beta)}\} = 2\delta_{\alpha\beta} J_0^{(\alpha)},\qquad \forall \alpha,\beta=1,2,3,4. \label{12dimn4a}
\end{align}
We use the ordered coordinate set
\begin{equation}
\tilde{J}=\{J_{-}^{(1)},J_{+}^{(1)},J_{0}^{(1)},J_{-}^{(2)},J_{+}^{(2)},J_{0}^{(2)},J_{-}^{(3)},J_{+}^{(3)},J_{0}^{(3)},J_{-}^{(4)},J_{+}^{(4)},J_{0}^{(4)}\} \label{12dimstr}
\end{equation}
with associated Poisson brackets
\begin{equation}
\{\tilde{J}_i,\tilde{J}_j\}= \tilde{C}^k_{ij} \tilde{J}_k  \label{12dimn4b}
\end{equation}
where the structure constants $\tilde{C}^k_{ij}$ are determined by \cref{12dimn4a}.  Throughout this $N=4$ subsection, indices occurring in the same multi-index are pairwise distinct; whenever $i,j,k,\ell$ occur in the same formula, they denote pairwise distinct elements of $\{1,2,3,4\}$.  The coproduct Casimirs are the partial Casimirs
\begin{align}
    \mathcal P_4&=\{C_i,\,C_{ij},\,C_{ijk},\,C_{1234}\},\qquad i,j,k=1,2,3,4 \label{gen4na}
\end{align}
Here, the total Casimir $C_{1234}$ is given by
\begin{align}
    C_{1234}&=\sum C_{ij}-2\sum C_i =C_{12}+C_{13}+C_{14}+C_{23}+C_{24}+C_{34}-2(C_1+C_2+C_3+C_4).
\end{align}
All elements are functions on the twelve-generator space \cref{12dimstr} with bracket \cref{12dimn4b}.  The single-site Casimirs and $C_{1234}$ are central.  Noncommuting two-site Casimirs generate the elements
\begin{align}
    \cF_{ij,jk}&:=\frac{1}{2}\{C_{ij},C_{jk}\}.\label{Fijk}
\end{align}
For any three distinct sites, the undeformed three-site Casimir satisfies
\begin{align}
    C_{ijk}&=C_{ij}+C_{ik}+C_{jk}-C_i-C_j-C_k.\label{eq:threesite-general}
\end{align}
A symmetric generating set is
\begin{align}
    \cK&=\{C_i,\,C_{ij},\,C_{ijk},\,C_{1234}, \, \cF_{ij,jk} \},\qquad i,j,k=1,2,3,4.  \label{gen4nb}
\end{align}
Although this set is algebraically dependent, it gives a symmetric presentation.  All relations below are obtained directly in the twelve-generator Poisson algebra and then rewritten in terms of these generators.

Using \cref{eq:threesite-general} and \cref{Fijk}, brackets of overlapping three-site Casimirs reduce to the $F$ generators:
\begin{align} 
\{C_{ijk},C_{jk\ell}\}
= 2\cF_{ik,k\ell}+2\cF_{ij,j\ell}.
\end{align}
Direct computation gives the quadratic algebra generated by the six $C_{ij}$ and the corresponding $\cF_{ij,jk}$:
\begin{align}
    \{\cF_{jk,k\ell},C_{ij}\}&=C_{ik}C_{j\ell}-C_{i\ell}C_{jk}-(C_i-C_j)(C_k-C_\ell)+(C_j+C_k)C_{i\ell}+(C_i+C_\ell)C_{jk}\non\\&\quad-(C_j+C_\ell)C_{ik}-(C_i+C_k)C_{j\ell},\\
    \{\cF_{ij,jk},C_{ij}\}&=C_{ij}C_{ik}-C_{ij}C_{jk}-(C_i-C_j)(C_{ijk}-C_k).
\end{align}
For a fixed triple $(i,j,k)$, eliminating the remaining two-site Casimir gives
\begin{align}
    \{\cF_{ij,jk},C_{ij}\}&=-C_{ij}^2-2C_{ij}C_{jk}+(C_{ijk}+C_i+C_j+C_k)C_{ij}-(C_i-C_j)(C_{ijk}-C_k)
\end{align}
This yields four overlapping copies of the rank-one Racah algebra, one for each fixed triple $(i,j,k)$.

\section[Non-standard slz(2) deformation: the three-site algebra]{Non-standard $\mathfrak{sl}_z(2)$ deformation: the three-site algebra}

Let us consider the non-standard $\mathfrak{sl}_z(2,\RR)$ Poisson coalgebra defined by the following deformed Poisson brackets and coproduct map.  Its undeformed limit is equivalent to the Lie--Poisson algebra of Section~3 after a simple rescaling and sign convention; all coefficients in the present section follow the normalization displayed below.
\begin{align}
    \{J_0,J_+\}=2J_+\cosh zJ_-,\qquad \{J_0,J_-\}=-2\frac{\sinh zJ_-}{z},\qquad \{J_-,J_+\}=4J_0,\label{nonstandrelations}
\end{align}
\begin{align}
    \D_z^{(2)}(J_-)=J_-^{(1)}+J_-^{(2)},\qquad \D_z^{(2)}(J_i)=J_i\otimes e^{zJ_-}+e^{-zJ_-}\otimes J_i,\quad i=+,0\label{nonstandcoprod}
\end{align}

For three copies, let $\{J_i^{(\alpha)}\}$ denote the generators at sites $\alpha=1,2,3$, with $i=\pm,0$.  The corresponding nine-dimensional Poisson structure is
\begin{align}
    \{J_0^{(\alpha)},J_+^{(\beta)}\}=2 \delta_{\alpha \beta} J_+^{(\alpha)} \cosh zJ_-^{(\alpha)},\qquad \{J_0^{(\alpha)},J_-^{(\beta)}\}=-2 \delta_{\alpha \beta} \frac{\sinh zJ_-^{(\alpha)}}{z},\qquad \{J_-^{(\alpha)},J_+^{(\beta)}\}=4 \delta_{\alpha \beta} J_0^{(\alpha)}. \label{z9dimn3a}
\end{align}
We will use the nine-generator coordinate set
\begin{equation}
\tilde{J}=\{J_{-}^{(1)},J_{+}^{(1)},J_{0}^{(1)},J_{-}^{(2)},J_{+}^{(2)},J_{0}^{(2)},J_{-}^{(3)},J_{+}^{(3)},J_{0}^{(3)}\} \label{z9dimstr}
\end{equation}
and related 9-dimensional Poisson algebra
\begin{equation}
\{\tilde{J}_i,\tilde{J}_j\}= g_{ij}(\tilde{J}),  \label{z9dimn3b}
\end{equation}
where the functions $g_{ij}(\tilde{J})$ are obtained from \cref{z9dimn3a}. The Casimir function for $\mathfrak{sl}_z(2,\RR)$ is given by
\begin{align}
    \cC_z=\frac{\sinh zJ_-}{z}J_+-J_0^2.
\end{align}
For three sites the relevant connected Casimir functions are
\begin{align}
    C_{12}&:=\D_z^{(2)}(\cC_z)=\frac{\sinh[z(J_-^{(1)}+J_-^{(2)})]}{z}(J_+^{(1)}e^{zJ_-^{(2)}}+J_+^{(2)}e^{-zJ_-^{(1)}})-(J_0^{(1)}e^{zJ_-^{(2)}}+J_0^{(2)}e^{-zJ_-^{(1)}})^2,\\
    C_{23}&:=(\D_z)_R^{(2)}(\cC_z)=\frac{\sinh[z(J_-^{(2)}+J_-^{(3)})]}{z}(J_+^{(2)}e^{zJ_-^{(3)}}+J_+^{(3)}e^{-zJ_-^{(2)}})-(J_0^{(2)}e^{zJ_-^{(3)}}+J_0^{(3)}e^{-zJ_-^{(2)}})^2,\\
    C_{123}&:=\D_z^{(3)}(\cC_z)=\frac{\sinh[z(J_-^{(1)}+J_-^{(2)}+J_-^{(3)})]}{z}(J_+^{(1)}e^{z(J_-^{(2)}+J_-^{(3)})}+J_+^{(2)}e^{-z(J_-^{(1)}-J_-^{(3)})}+J_+^{(3)}e^{-z(J_-^{(1)}+J_-^{(2)})})\non\\
    &\qquad\qquad\qquad\qquad-(J_0^{(1)}e^{z(J_-^{(2)}+J_-^{(3)})}+J_0^{(2)}e^{-z(J_-^{(1)}-J_-^{(3)})}+J_0^{(3)}e^{-z(J_-^{(1)}+J_-^{(2)})})^2.
\end{align}

The deformed expressions contain both polynomial terms and exponentials of the elementary generators. The Poisson bracket remains well defined for these functions, although the search for closure must account for the exponential factors rather than relying only on polynomial degree.

The two connected two-site Casimirs do not Poisson-commute.  From this point onward, for the non-standard $q$-Poisson algebra, we use the unscaled normalization $\cF=\mathcal B$ of \cref{eq:Fnotation}; hence
\begin{align}
    \{C_{12},C_{23}\}&:=\cF_{12,23}.
\end{align}
The resulting generating set is
\begin{align}
    \cK&=\{C_1,C_2,C_3,C_{12},C_{23},C_{123},\cF_{12,23}\}, \label{zgen3na}
\end{align}
where $C_1,C_2,C_3,C_{123}$ are central.  The nonzero closure relations are
\begin{align}
    \{C_{12},\cF_{12,23}\}&=-8C_{12}^2-16C_{12}C_{23}+8(C_3-C_{123})(C_1-C_2)+8C_{12}(C_{123}+C_1+C_2+C_3),\\
    \{C_{23},\cF_{12,23}\}&=+8C_{23}^2+16C_{12}C_{23}+8(C_1-C_{123})(C_2-C_3)-8C_{23}(C_{123}+C_1+C_2+C_3).
\end{align}
These identities follow by direct evaluation in the nine-generator Poisson algebra \cref{z9dimstr}--\cref{z9dimn3b}.  A notable feature is that no explicit $z$ remains in their coefficients.  This is stronger than merely retaining a Racah-type shape: the deformation is carried by the partial Casimirs, which are transcendental functions of the elementary generators, while the closure polynomial itself is deformation-parameter independent.

The relation with the undeformed normalization can be made explicit.  Denote the $z\to0$ non-standard generators by $J_a^{\rm ns}$ and the undeformed generators of Section~3 by $J_a^{0}$.  The fixed change of normalization
\begin{equation}
 J_0^{\rm ns}=2J_0^{0},\qquad J_+^{\rm ns}=2J_+^{0},\qquad J_-^{\rm ns}=-2J_-^{0}
 \label{eq:nonstandard-normalization-limit}
\end{equation}
maps the limiting non-standard brackets to the undeformed brackets of Section~3 and gives
\begin{equation}
 C_A^{\rm ns}\longrightarrow -4C_A^{0},\qquad
 \cF_{12,23}^{\rm ns}\longrightarrow32\cF^{0},
 \label{eq:nonstandard-CF-limit}
\end{equation}
where $\cF^{0}=\tfrac12\{C_{12}^{0},C_{23}^{0}\}$. Consequently,
\begin{equation}
 \{C_{12}^{\rm ns},\cF_{12,23}^{\rm ns}\}
 \longrightarrow-128\{C_{12}^{0},\cF^{0}\},
 \label{eq:nonstandard-bracket-limit}
\end{equation}
and, after division by this common normalization factor, the two displayed deformed relations reduce term by term to \cref{classracah1,classracah2}.  The numerical factor $-128$ is a consequence of the different Casimir and bracket-generator normalizations used in Sections~3 and 4.

\section[Non-standard slz(2) deformation: the four-site algebra]{Non-standard $\mathfrak{sl}_z(2)$ deformation: the four-site algebra}

For four copies, with generators $\{J_i^{(\alpha)}\}$ for $\alpha=1,2,3,4$ and $i=\pm,0$, the Poisson structure is
\begin{align}
    \{J_0^{(\alpha)},J_+^{(\beta)}\}=2 \delta_{\alpha \beta} J_+^{(\alpha)} \cosh zJ_-^{(\alpha)},\qquad \{J_0^{(\alpha)},J_-^{(\beta)}\}=-2 \delta_{\alpha \beta} \frac{\sinh zJ_-^{(\alpha)}}{z},\qquad \{J_-^{(\alpha)},J_+^{(\beta)}\}=4 \delta_{\alpha \beta} J_0^{(\alpha)}. \label{z12dimn4a}
\end{align}
We use the ordered twelve-generator set
\begin{equation}
\tilde{J}=\{J_{-}^{(1)},J_{+}^{(1)},J_{0}^{(1)},J_{-}^{(2)},J_{+}^{(2)},J_{0}^{(2)},J_{-}^{(3)},J_{+}^{(3)},J_{0}^{(3)}, J_{-}^{(4)},J_{+}^{(4)},J_{0}^{(4)}  \} \label{z12dimstr}
\end{equation}
with associated twelve-dimensional Poisson brackets
\begin{equation}
\{\tilde{J}_i,\tilde{J}_j\}= g_{ij}(\tilde{J}),  \label{z12dimn4b}
\end{equation}
where $g_{ij}$ is determined by \cref{z12dimn4a}; the coproduct remains \cref{nonstandcoprod}.

In addition to $C_{12},C_{23},C_{123}$ from the three-site construction, define
\begin{align}
    C_{34}&:=(\D_z)_R^{(2)}({\cC_z}),\qquad C_{234}:=(\D_z)_R^{(3)}({\cC_z}),\qquad C_{1234}:=\D_z^{(4)}({\cC_z}).
\end{align}
The four-site closure is summarized in the following proposition.

\begin{proposition}\label{resultsprop}
    The symmetry algebra $\cS_4$ from $\mathfrak{sl}_z(2)$ coalgebra with $N=4$ is generated by nine non-central elements
    \begin{align}      C_{12},\,C_{23},\,C_{34},\,C_{123},\,C_{234},\,\cF_{12,23},\,\cF_{23,34},\,\cF_{12,234},\,\cF_{34,123},
    \end{align}
    and five central elements
    \begin{align}
        C_1,\,C_2,\,C_3,\,C_4,\,C_{1234}.
    \end{align}
    The bracket generators are defined by
    \begin{align}
    \{C_{12},C_{23}\}&:=\cF_{12,23}, \quad \{C_{23},C_{34}\}:=\cF_{23,34}, \quad
    \{C_{12},C_{234}\}:=\cF_{12,234}, \quad \{C_{34},C_{123}\}:=\cF_{34,123}.
\end{align}
The algebra closes quadratically through the following relations.
\begin{align}
    \{C_{12},\cF_{12,23}\}&=8C_{123}C_{12}-8C_{12}^2-16C_{12}C_{23}-8C_{123}(C_1-C_2)\non\\
    &\quad+8(C_1-C_2)C_3+8C_{12}(C_1+C_2+C_3),\\
    \{C_{12},\cF_{23,34}\}&=8C_{123}C_{234}-8C_{12}C_{234}-8C_{1234}C_{23}+8C_{12}C_{23}-8C_{123}C_{34}+8C_{12}C_{34}+8C_{23}C_{34}\non\\
    &\quad-8C_{123}C_2-8C_{34}C_2-8C_{12}C_3-8C_{234}C_3+8C_2C_3+8C_{1234}(C_2+C_3),\\
    \{C_{12},\cF_{12,234}\}&=8C_{1234}C_{12}-8C_{12}^2-16C_{12}C_{234}+8C_{12}C_{34}\non\\
    &\quad-8C_{1234}(C_1-C_2)+8C_{34}(C_1-C_2)+8C_{12}(C_1+C_2),\\
    \{C_{23},\cF_{12,23}\}&=-8C_{123}C_{23}+16C_{12}C_{23}+8C_{23}^2-8C_{123}(C_2-C_3)\non\\
    &\quad-8C_{23}(C_1+C_2+C_3)+8(C_1C_2-C_1C_3),\\
    \{C_{23},\cF_{23,34}\}&=8C_{234}C_{23}-8C_{23}^2-16C_{23}C_{34}-8C_{234}(C_2-C_3)\non\\
    &\quad+8C_{23}(C_2+C_3+C_4)+8(C_2C_4-C_3C_4),\\
    \{C_{23},\cF_{12,234}\}&=-8C_{123}C_{234}+8C_{12}C_{234}+8C_{12}C_{23}+8C_{234}C_{23}+8C_{123}C_{34}-8C_{23}C_{34}\non\\
    &\quad-8C_{23}C_1-8C_{34}C_1-8C_{123}C_2-8C_{234}C_2-8C_{12}C_4+8(C_1C_2+C_1C_4+C_2C_4),\\
    \{C_{23},\cF_{34,123}\}&=-8C_{123}C_{234}+8C_{12}C_{234}+8C_{123}C_{23}-8C_{12}C_{23}+8C_{123}C_{34}+8C_{23}C_{34}\non\\
    &\quad-8C_{34}C_1-8C_{12}C_4-8C_{123}C_3-8C_{234}C_3-8C_{23}C_4+8(C_1C_3+C_1C_4+C_3C_4),
\end{align}
\begin{align}
    \{C_{34},\cF_{12,23}\}&=-8C_{123}C_{234}+8C_{12}C_{234}+8C_{1234}C_{23}-8C_{12}C_{23}+8C_{123}C_{34}-8C_{12}C_{34}\non\\
    &\quad-8C_{23}C_{34}+8C_{123}C_2+8C_{34}C_2+8C_{12}C_3+8C_{234}C_3\non\\
    &\quad-8C_2C_3-8C_{1234}(C_2+C_3),\\
    \{C_{34},\cF_{23,34}\}&=-8C_{234}C_{34}+16C_{23}C_{34}+8C_{34}^2-8C_{234}(C_3-C_4)+8C_2(C_3-C_4)\non\\
    &\quad-8C_{34}(C_2+C_3+C_4),\\
    \{C_{34},\cF_{34,123}\}&=8C_{1234}C_{34}-16C_{123}C_{34}+8C_{12}C_{34}-8C_{34}^2\non\\
    &\quad+8C_{1234}(C_3-C_4)-8C_{12}(C_3-C_4)+8C_{34}(C_3+C_4),
\end{align}
\begin{align}
    \{C_{123},C_{234}\}&=\cF_{12,234}-\cF_{34,123}-\cF_{12,23}-\cF_{23,34},\\
    \{C_{123},\cF_{23,34}\}&=8C_{123}C_{234}-8C_{12}C_{234}-8C_{123}C_{23}+8C_{12}C_{23}-8C_{123}C_{34}-8C_{23}C_{34}\non\\
    &\quad+8C_{34}C_1+8C_{12}C_4+8C_{123}C_3+8C_{234}C_3+8C_{23}C_4\non\\
    &\quad-8(C_1C_3+C_1C_4+C_3C_4),\\
    \{C_{123},\cF_{12,234}\}&=8C_{1234}C_{12}-8C_{123}C_{234}-8C_{123}C_{12}-8C_{12}C_{234}+8C_{1234}C_{23}+8C_{12}C_{23}\non\\
    &\quad+8C_{123}C_{34}-8C_{23}C_{34}+8C_{123}C_1+8C_{34}C_1+8C_{234}C_3\non\\
    &\quad-8C_1C_3-8C_{1234}(C_1+C_3),\\
    \{C_{123},\cF_{34,123}\}&=-8C_{1234}C_{123}+8C_{123}^2+8C_{1234}C_{12}-8C_{123}C_{12}+16C_{123}C_{34}-8C_{1234}C_3\non\\
    &\quad-8C_{12}C_4+8C_3C_4-8C_{123}(C_3+C_4),
\end{align}
\begin{align}
    \{C_{234},\cF_{12,23}\}&=-8C_{123}C_{234}+8C_{12}C_{234}+8C_{12}C_{23}+8C_{234}C_{23}+8C_{123}C_{34}-8C_{23}C_{34}\non\\
    &\quad-8C_{23}C_1-8C_{34}C_1-8C_{123}C_2-8C_{234}C_2-8C_{12}C_4+8(C_1C_2+C_1C_4+C_2C_4),\\
    \{C_{234},\cF_{12,234}\}&=-8C_{1234}C_{234}+16C_{12}C_{234}+8C_{234}^2+8C_{1234}C_{34}-8C_{234}C_{34}-8C_{34}C_1\non\\
    &\quad-8C_{1234}C_2+8C_1C_2-8C_{234}(C_1+C_2),\\
    \{C_{234},\cF_{34,123}\}&=-8C_{123}C_{234}+8C_{12}C_{234}+8C_{1234}C_{23}-8C_{12}C_{23}+8C_{1234}C_{34}-8C_{123}C_{34}\non\\
    &\quad-8C_{234}C_{34}+8C_{23}C_{34}+8C_{123}C_2+8C_{12}C_4+8C_{234}C_4\non\\
    &\quad-8C_2C_4-8C_{1234}C_2-8C_{1234}C_4,\\
    \{\cF_{12,23},\cF_{23,34}\}&=-8(\cF_{12,23}+\cF_{23,34}-\cF_{12,234}+\cF_{34,123})(C_{23}-C_2-C_3),\\
    \{\cF_{12,23},\cF_{12,234}\}&=8\cF_{34,123}(C_{12}-C_1-C_2),\\
    \{\cF_{12,23},\cF_{34,123}\}&=8\cF_{12,23}(C_{123}-C_{12}+C_{34}-C_4)-8\cF_{12,234}(C_{123}-C_{12}+C_3),\\
    \{\cF_{23,34},\cF_{12,234}\}&=8\cF_{23,34}(C_{12}+C_{234}-C_{34}-C_1)+8\cF_{34,123}(C_{234}-C_{34}+C_2),\\
    \{\cF_{23,34},\cF_{34,123}\}&=-8\cF_{12,234}(C_{34}-C_3-C_4),\\
    \{\cF_{12,234},\cF_{34,123}\}&=8(2\cF_{23,34}+\cF_{34,123})C_{12}
    +8(2\cF_{12,23}-\cF_{12,234})C_{34}\non\\
    &\quad-8\cF_{34,123}(C_1-C_2)-8\cF_{12,234}(C_3-C_4).
\end{align}
In addition to the seed--seed zero brackets already forced by nested/disjoint support in Section~2, the following omitted seed--bracket-generator brackets vanish identically:
\begin{align}
 \{C_{12},\cF_{34,123}\}=0,\qquad
 \{C_{34},\cF_{12,234}\}=0,\qquad
 \{C_{123},\cF_{12,23}\}=0,\qquad
 \{C_{234},\cF_{23,34}\}=0.
 \label{eq:omittedzeros}
\end{align}
\end{proposition}

\begin{proof}
    The result follows by direct evaluation of the Poisson brackets using the coproduct structure \cref{nonstandcoprod} and the twelve-generator Poisson bracket \cref{z12dimn4a}. Every bracket among the displayed generators reduces to a polynomial of degree at most two in the generators of $\cS_4$. The one-site Casimirs and $C_{1234}$ are central by construction. The Jacobi identity gives useful consistency checks, including
    $\{C_{12},\cF_{23,34}\}=-\{C_{34},\cF_{12,23}\}$ and
    $\{C_{123},\cF_{23,34}\}=-\{C_{23},\cF_{34,123}\}$.

\end{proof}

A second structural feature is visible directly in the complete four-site table: despite the exponential dependence of the partial Casimirs on $z$, none of the closure coefficients contains $z$ explicitly.  Thus the deformation-parameter independence observed at three sites persists at rank two.  This statement concerns the polynomial relations among the chosen partial Casimirs and bracket generators; it does not mean that the underlying Poisson coalgebra or its coproduct is undeformed.

Because these relations are identities in the abstract Poisson coalgebra, they remain valid under any Poisson realization of the same non-standard $\mathfrak{sl}_z(2)$ algebra.  Section~6 gives one explicit realization.

\section{Hamiltonians, superintegrability and symplectic realizations}
The preceding sections determine the symmetry algebra before any realization is chosen.  We now explain in more detail how this abstract result enters the construction of Hamiltonian systems.  Two related viewpoints are useful: one can build the Hamiltonian from central Casimir data, or directly as a function of the total coproduct generators.  The second viewpoint is the one most commonly used in explicit coalgebra models, while the first makes the central structure of the symmetry algebra transparent.

For $N=3$, both in the undeformed and deformed problems, the one-site Casimirs together with the total Casimir form the Abelian subalgebra
\begin{equation}
 \mathcal A_3=\langle C_1,C_2,C_3,C_{123}\rangle,
\end{equation}
whereas for $N=4$ one has
\begin{equation}
 \mathcal A_4=\langle C_1,C_2,C_3,C_4,C_{1234}\rangle.
\end{equation}
The remaining left and right partial Casimirs are nontrivial integrals for Hamiltonians constructed from this central data.

\begin{proposition}
Let
\begin{equation}
 \mathcal A_N=\langle C_1,\ldots,C_N,C_{1\cdots N}\rangle.
\end{equation}
Every element of $\mathcal A_N$ Poisson-commutes with the left and right partial Casimirs generated by the coalgebra. Consequently, a Hamiltonian of the form
\begin{equation}
 H=g(\mathcal A_N)
\end{equation}
admits these partial Casimirs as integrals of motion.
\end{proposition}

\begin{proof}
The one-site Casimirs and the total Casimir mutually Poisson commute.  Equation~\eqref{poissoncommutegens} implies that each coalgebra-generated left or right partial Casimir Poisson-commutes with the total Casimir and with the one-site Casimirs.  The result then follows from the derivation property of the Poisson bracket.
\end{proof}

An equivalent and often more useful construction starts from the total coproduct generators.
\begin{proposition}\label{hamprop2}
Any Hamiltonian of the form
\begin{equation}
H=g\!\left(\Delta^{(N)}(J_-),\Delta^{(N)}(J_+),\Delta^{(N)}(J_0)\right)
\end{equation}
or
\begin{equation}
H_z=g_z\!\left(\Delta_z^{(N)}(J_-),\Delta_z^{(N)}(J_+),\Delta_z^{(N)}(J_0)\right)
\end{equation}
admits the corresponding left and right partial Casimirs as integrals of motion.
\end{proposition}

\begin{proof}
The coproduct is a Poisson-algebra homomorphism, so the total coproduct generators satisfy the same Poisson relations as the one-site generators.  Each coalgebra-generated partial Casimir Poisson-commutes with these total generators and hence with any smooth function of them.
\end{proof}

These propositions establish the existence of a large family of conserved quantities but do not determine their non-Abelian closure.  The latter is a separate problem.  Even for polynomial integrals, naive degree counting does not in general predict whether the brackets close at the expected degree, and higher-order superintegrable systems can require additional generators and nested relations \cite{ian2010,parr2025}.  The calculations in Sections~3--5 provide information that is not contained in the standard coalgebra superintegrability theorem: they identify a finite set of generators and give all of their Poisson brackets explicitly.

This distinction is particularly useful for the present deformation.  The elementary expressions contain exponentials and hyperbolic functions, yet the algebra generated by the partial Casimirs closes quadratically.  Since this was established before choosing a realization, the same quadratic relations hold for every Poisson realization of $\mathfrak{sl}_z(2)$.  A realization may introduce additional algebraic or functional relations among the generators, but it cannot invalidate the universal identities obtained at the Poisson-coalgebra level.

The coalgebra-generated quantities should not be assumed to exhaust the symmetry of every particular Hamiltonian.  Special choices of the function $g$ or $g_z$ may admit additional integrals and become maximally superintegrable.  The Smorodinsky--Winternitz and generalized Kepler--Coulomb families provide familiar examples in which the full symmetry algebra contains generators beyond the universal coalgebra sector \cite{ques11,bh09kc,lat21}.  The algebra constructed here is best viewed as a symmetry subalgebra, independent of the chosen realization, that is common to the full class of Hamiltonians generated from the same coalgebra.

\subsection{Symplectic realization and explicit Hamiltonians}
Symplectic realizations of the non-standard deformed $\mathfrak{sl}(2)$ Poisson algebra have been used extensively in the coalgebra approach to integrable systems on curved spaces \cite{bal05,bal09b,bal07,bal08}.  A realization $D_z$ must preserve the Poisson brackets; for physical applications it is also useful that the undeformed limit be regular and that the realized generators have a standard dependence on the momenta.

We use the one-particle realization
\begin{align}
    D_z(J_-)&=q_1^2,\qquad
    D_z(J_+)=\frac{\sinh(zq_1^2)}{zq_1^2}p_1^2+\frac{zb_1}{\sinh(zq_1^2)},\qquad
    D_z(J_0)=\frac{\sinh(zq_1^2)}{zq_1}p_1,
    \label{oneparticlerealization}
\end{align}
for which the one-site Casimir at site $i$ takes the constant value $D_z(C_i)=b_i$.  Applying $D_z$ tensorwise to the iterated coproduct gives
\begin{equation}
    J_i^{(N)}=D_z^{\otimes N}\!\left(\Delta_z^{(N)}(J_i)\right),\qquad i=+,-,0.
\end{equation}
The Poisson-homomorphism property ensures that every realized partial Casimir satisfies exactly the same abstract relations derived in Sections~4 and 5.

In canonical variables $(q_i,p_i)$ the $N$-particle generators are
\begin{align}
    J_-^{(N)}&=\sum^N_{i=1}q_i^2,\\
    J_+^{(N)}&=\sum_{i=1}^N\left(\frac{\sinh(zq_i^2)}{zq_i^2}p_i^2+\frac{zb_i}{\sinh(zq_i^2)}\right)
    \exp\!\left(-z\sum_{k=1}^{i-1}q_k^2+z\sum_{l=i+1}^{N}q_l^2\right),\\
    J_0^{(N)}&=\sum_{i=1}^N\left(\frac{\sinh(zq_i^2)}{zq_i}p_i\right)
    \exp\!\left(-z\sum_{k=1}^{i-1}q_k^2+z\sum_{l=i+1}^{N}q_l^2\right).
\end{align}
Applying the realization to the left and right partial Casimirs yields the standard $(2N-3)$ coalgebra family
\begin{align}
    C^{(m)}&=\sum_{1\leq i<j\leq m} I^z_{ij}
    \exp\!\left(-2z\sum_{k=1}^{i-1}q_k^2-zq_i^2+zq_j^2+2z\sum_{l=j+1}^{m}q_l^2\right)\non\\
    &\quad+\sum_{i=1}^{m}b_i
    \exp\!\left(-2z\sum_{k=1}^{i-1}q_k^2+2z\sum_{l=i+1}^{m}q_l^2\right),\label{eq:leftrealizedgeneral}\\
    C_{(m)}&=\sum_{N-m+1\leq i<j\leq N} I^z_{ij}
    \exp\!\left(-2z\sum_{k=N-m+1}^{i-1}q_k^2-zq_i^2+zq_j^2+2z\sum_{l=j+1}^{N}q_l^2\right)\non\\
    &\quad+\sum_{i=N-m+1}^{N}b_i
    \exp\!\left(-2z\sum_{k=N-m+1}^{i-1}q_k^2+2z\sum_{l=i+1}^{N}q_l^2\right),\label{eq:rightrealizedgeneral}
\end{align}
where $m=2,\ldots,N$.  In the second line of each expression the upper-tail sum starts at $l=i+1$, so every index remains within the stated range.  The pair kernel is
\begin{equation}
\begin{aligned}
 I^z_{ij}={}&
 \frac{\sinh(zq_i^2)\sinh(zq_j^2)}{z^2q_i^2q_j^2}
 (q_ip_j-q_jp_i)^2
 +b_i\frac{\sinh(zq_j^2)}{\sinh(zq_i^2)}
 +b_j\frac{\sinh(zq_i^2)}{\sinh(zq_j^2)}.
\end{aligned}
\label{eq:pairkernel}
\end{equation}

For the two-site partial Casimirs this gives the compact exact expressions
\begin{align}
 L_{12}&:=D_z(C_{12})
 =b_1e^{2zq_2^2}+b_2e^{-2zq_1^2}
 +I^z_{12}e^{-z(q_1^2-q_2^2)},\label{eq:L12compact}\\
 L_{23}&:=D_z(C_{23})
 =b_2e^{2zq_3^2}+b_3e^{-2zq_2^2}
 +I^z_{23}e^{-z(q_2^2-q_3^2)},\label{eq:L23compact}\\
 L_{34}&:=D_z(C_{34})
 =b_3e^{2zq_4^2}+b_4e^{-2zq_3^2}
 +I^z_{34}e^{-z(q_3^2-q_4^2)}.\label{eq:L34compact}
\end{align}
The kernel form is preferable to repeatedly writing the same three-line phase-space expression.  The expression is completely explicit, and Appendix~B gives the corresponding factorized formulae for the four realized bracket generators without leaving any Poisson derivatives unevaluated.

The abstract construction produces families of integrable Hamiltonians.  Three standard examples are:
\begin{itemize}
    \item the geodesic Hamiltonian on a space of non-constant curvature,
    \begin{equation}
        H_z=\frac{1}{2}J_+^{(N)};
    \end{equation}
    \item the deformed oscillator,
    \begin{equation}
        H_z=\frac{1}{2}J_+^{(N)}+\omega^2J_-^{(N)};
    \end{equation}
    \item the deformed Kepler--Coulomb Hamiltonian,
    \begin{equation}
        H_z=\frac{1}{2}J_+^{(N)}-\frac{k}{\sqrt{J_-^{(N)}}}.
    \end{equation}
\end{itemize}
In each case the same left and right partial Casimirs are conserved because the Hamiltonian is a function of total coproduct generators.  The deformation parameter controls the departure from the undeformed geometry, while the algebra derived in Section~5 controls the universal four-site coalgebra sector of the symmetry.  Further Hamiltonians can be obtained by replacing the displayed functions by more general $g_z(J_-^{(N)},J_+^{(N)},J_0^{(N)})$ without altering this universal set of integrals.

\subsection[A maximally superintegrable q-deformed Smorodinsky--Winternitz model]{A maximally superintegrable $q$-deformed Smorodinsky--Winternitz model}
\label{sec:qswapplication}

The preceding examples illustrate the universal coalgebra sector, but they do not by themselves address the distinction between quasi-maximal and maximal superintegrability.  This distinction can be made explicit by specializing the Hamiltonian while retaining the same non-standard $\mathfrak{sl}_z(2)$ Poisson coalgebra.  A particularly useful example is the St\"ackel deformation of the Smorodinsky--Winternitz (SW) system constructed in \cite{bal04}.  In the normalization used in that reference the Hamiltonian function on the abstract Poisson coalgebra is
\begin{equation}
 \mathcal H^{\rm SSW}_z
 =
 J_+ e^{zJ_-}
 +\frac{\omega^2}{2z}\left(e^{2zJ_-}-1\right).
 \label{eq:qSWHabstract}
\end{equation}
The sign convention for the Casimir in \cite{bal04} is opposite to the convention adopted in the present paper; this does not affect the Hamiltonian or the St\"ackel integrals below.  Under the realization \eqref{oneparticlerealization}, the $N$-particle Hamiltonian is
\begin{align}
 H^{\rm SSW}_{z,N}
 ={}&
 \sum_{i=1}^N
 \frac{\sinh(zq_i^2)}{zq_i^2}
 e^{zq_i^2}
 \exp\!\left(2z\sum_{k=i+1}^Nq_k^2\right)
 \left[
 p_i^2+b_i\left(\frac{zq_i}{\sinh(zq_i^2)}\right)^2
 \right]
 +\frac{\omega^2}{2z}
 \left[
 \exp\!\left(2z\sum_{j=1}^Nq_j^2\right)-1
 \right].
 \label{eq:qSWHN}
\end{align}
This Hamiltonian is of St\"ackel type.  Besides the left and right coalgebra integrals, it admits the additional commuting family \cite{bal04}
\begin{align}
 Z_j={}&
 \sum_{i=1}^{j-1}
 \frac{\sinh(zq_i^2)}{2zq_i^2}
 e^{zq_i^2}
 \exp\!\left(2z\sum_{k=i+1}^{j-1}q_k^2\right)
 \left[
 p_i^2+b_i\left(\frac{zq_i}{\sinh(zq_i^2)}\right)^2
 \right]
 \notag\\
 &+
 \frac{\omega^2}{4z}
 \left[
 \exp\!\left(2z\sum_{k=1}^{j-1}q_k^2\right)-1
 \right],
 \qquad j=2,\ldots,N.
 \label{eq:qSWZj}
\end{align}
In particular, $Z_2$ can be used as the one additional integral that is not supplied by the coalgebra construction.  The St\"ackel system is maximally superintegrable: the $2N-3$ universal coalgebra integrals, together with the Hamiltonian, give $2N-2$ functionally independent constants, and $Z_2$ completes this family to $2N-1$ \cite{bal04}.

The algebra needed below is already fixed before choosing a symplectic realization.  We isolate the relevant two-block identity once and then use it repeatedly.

\begin{lemma}[Exact abstract two-block St\"ackel identity]\label{lem:qSWtwoblock}
Let $(J_-^{(1)},J_+^{(1)},J_0^{(1)})$ and $(J_-^{(2)},J_+^{(2)},J_0^{(2)})$ be two mutually Poisson-commuting copies of the non-standard algebra \eqref{nonstandrelations}, with one-block Casimirs
\begin{equation}
 a=\cC_z^{(1)},\qquad b=\cC_z^{(2)}.
\end{equation}
Let $J_\mu^{(12)}=\Delta_z(J_\mu)$, $\mu=-,+,0$, be their two-block coproduct generators and define
\begin{equation}
 h_z(J_-,J_+,J_0):=
 \frac12\left[J_+e^{zJ_-}+\frac{\omega^2}{2z}\left(e^{2zJ_-}-1\right)\right],
 \label{eq:qSWhzhalf}
\end{equation}
so that $h_z=\mathcal H_z^{\rm SSW}/2$.  Set
\begin{equation}
 A=\cC_z(J^{(12)}),\qquad
 X=h_z(J^{(1)}),\qquad
 T=h_z(J^{(12)}),\qquad
 G=\{A,X\}.
 \label{eq:qSWtwoblockdefs}
\end{equation}
Then $a,b,T$ are central in the Poisson algebra generated by $A,X,G$, and
\begin{align}
 \{A,G\}={}&-8zA^2-16AX+8AT+8z(a+b)A+8(a-b)T,
 \label{eq:qSWtwoblockAG}\\
 \{X,G\}={}&8X^2+16zAX-8XT+4\omega^2A
 -8z(a+b)X-8zaT
 +8z^2ab-4\omega^2(a+b).
 \label{eq:qSWtwoblockXG}
\end{align}
Moreover,
\begin{align}
 \cK_{2\rm b}={}&G^2+16zA^2X+16AX^2-16ATX
 -16z(a+b)AX-16(a-b)TX
 \notag\\
 &+4\omega^2A^2
 +\left[-16zaT+16z^2ab-8\omega^2(a+b)\right]A
 \label{eq:qSWtwoblockK}
\end{align}
has the central value
\begin{equation}
 \cK_{2\rm b}=-16aT^2+16za(b-a)T-4\omega^2(a-b)^2.
 \label{eq:qSWtwoblockKvalue}
\end{equation}
\end{lemma}

\begin{proof}
The statement is an identity of the abstract two-copy Poisson algebra.  Localize only with respect to the factors needed to solve the two Casimir equations and introduce
\begin{equation}
 s=e^{-zJ_-^{(1)}},\qquad t=e^{-zJ_-^{(2)}},\qquad
 u=J_0^{(1)},\qquad v=J_0^{(2)}.
 \label{eq:qSWstuv}
\end{equation}
The Casimir constraints give
\begin{equation}
 J_+^{(1)}=\frac{2zs(a+u^2)}{1-s^2},\qquad
 J_+^{(2)}=\frac{2zt(b+v^2)}{1-t^2},
 \label{eq:qSWJpluslocalized}
\end{equation}
and the only nonzero brackets among $s,t,u,v$ are
\begin{equation}
 \{s,u\}=s^2-1,\qquad \{t,v\}=t^2-1.
 \label{eq:qSWlocalizedbrackets}
\end{equation}
Hence
\begin{equation}
 \{f,g\}=(s^2-1)(f_sg_u-f_ug_s)
 +(t^2-1)(f_tg_v-f_vg_t).
 \label{eq:qSWlocalizedPB}
\end{equation}
The abstract quantities become the rational functions
\begin{align}
 X={}&\frac{z(a+u^2)}{1-s^2}
 +\frac{\omega^2(1-s^2)}{4zs^2},
 \label{eq:qSWlocalizedX}\\
 Y={}&\frac{z(b+v^2)}{1-t^2}
 +\frac{\omega^2(1-t^2)}{4zt^2},
 \qquad T=\frac{X}{t^2}+Y,
 \label{eq:qSWlocalizedT}\\
 A={}&\frac{(1-s^2t^2)(a+u^2)}{t^2(1-s^2)}
 +\frac{(1-s^2t^2)(b+v^2)}{1-t^2}
 -\left(\frac{u}{t}+sv\right)^2.
 \label{eq:qSWlocalizedA}
\end{align}
Equation~\eqref{eq:qSWlocalizedT} gives $\{X,T\}=0$ directly.  Since $A$ is the Casimir of the coproduct copy and $T$ is a function of the total coproduct generators, $\{A,T\}=0$; the Jacobi identity then gives $\{G,T\}=0$.

Because $X$ is independent of $t,v$, the definition of $G$ reduces to
\begin{equation}
 G=(s^2-1)(A_sX_u-A_uX_s).
 \label{eq:qSWlocalizedG}
\end{equation}
Substitution of \eqref{eq:qSWlocalizedX}--\eqref{eq:qSWlocalizedG} into \eqref{eq:qSWtwoblockK}, followed by reduction to the common denominator
$s^4t^4z^2(1-s^2)^2(1-t^2)^2$, gives identically the value \eqref{eq:qSWtwoblockKvalue}.  Equivalently,
\begin{equation}
 G^2=\Phi(A,X,T;a,b),
 \label{eq:qSWPhiidentity}
\end{equation}
where $\Phi$ is obtained by moving the non-$G^2$ terms of \eqref{eq:qSWtwoblockK} to the right and inserting \eqref{eq:qSWtwoblockKvalue}.  This is an exact rational identity in the localized abstract Poisson algebra.  More explicitly, the polynomial coordinate ring in $s,t,u,v,a,b,z,\omega$ is an integral domain, so localization at the multiplicative set generated by $z,s,t,1-s^2,1-t^2$ is injective.  Equivalently, multiplying by the displayed common denominator gives a numerator identity in the unlocalized coordinate ring.  Hence \eqref{eq:qSWPhiidentity} is not merely a generic-point statement and extends back after denominators are cleared.

The two bracket relations follow from this single identity.  Since $a,b,T$ are central and $\{A,X\}=G$,
\begin{equation}
 2G\{A,G\}=\Phi_XG,
 \qquad
 2G\{X,G\}=-\Phi_A G.
 \label{eq:qSWdifferentiatePhi}
\end{equation}
The localized rational-function algebra is an integral domain.  Moreover, $G$ is not the zero element: after writing \eqref{eq:qSWlocalizedG} over a common denominator, the coefficient of $b$ in its numerator is $4s^3t^4uz^2(1-s^2)$, which is not the zero polynomial.  Therefore the common factor $G$ in \eqref{eq:qSWdifferentiatePhi} can be cancelled as an identity.  Differentiating the explicit polynomial $\Phi$ gives exactly \eqref{eq:qSWtwoblockAG} and \eqref{eq:qSWtwoblockXG}.  Thus the closure and Casimir relation follow solely from the defining Poisson brackets, coproduct and Casimir, with no appeal to a canonical realization.
\end{proof}

For the three-dimensional case that is directly related to the rank-one algebra of Section~4, the first two St\"ackel integrals read
\begin{align}
 Z_2={}&
 \frac{\sinh(zq_1^2)}{2zq_1^2}e^{zq_1^2}
 \left[
 p_1^2+b_1\left(\frac{zq_1}{\sinh(zq_1^2)}\right)^2
 \right]
 +\frac{\omega^2}{4z}\left(e^{2zq_1^2}-1\right),
 \label{eq:qSWZ2}\\
 Z_3={}&
 \frac{\sinh(zq_1^2)}{2zq_1^2}e^{zq_1^2+2zq_2^2}
 \left[
 p_1^2+b_1\left(\frac{zq_1}{\sinh(zq_1^2)}\right)^2
 \right]
 \notag\\
 &+
 \frac{\sinh(zq_2^2)}{2zq_2^2}e^{zq_2^2}
 \left[
 p_2^2+b_2\left(\frac{zq_2}{\sinh(zq_2^2)}\right)^2
 \right]
 +\frac{\omega^2}{4z}
 \left(e^{2z(q_1^2+q_2^2)}-1\right).
 \label{eq:qSWZ3}
\end{align}
The undeformed limits are
\begin{align}
 Z_2&\longrightarrow
 \frac12\left(p_1^2+\frac{b_1}{q_1^2}+\omega^2q_1^2\right),\\
 Z_3&\longrightarrow
 \frac12\left(
 p_1^2+\frac{b_1}{q_1^2}
 +p_2^2+\frac{b_2}{q_2^2}
 +\omega^2(q_1^2+q_2^2)
 \right),
 \qquad z\to0,
 \label{eq:qSWundeformedZ}
\end{align}
so that the additional integral continuously deforms the familiar Cartesian SW separation constant.

The Hamiltonian \eqref{eq:qSWHabstract} and the St\"ackel family $Z_j$ are recalled from Ref.~\cite{bal04}; our calculation concerns their Poisson algebraic coupling to the partial Casimir $L_{12}$.  The extra integral does more than establish maximal superintegrability: it produces a quadratic algebra specific to the St\"ackel model that can be compared directly with the universal deformed Racah algebra.  Let
\begin{equation}
 A=L_{12},\qquad X=Z_2,\qquad T=Z_3,\qquad
 G=\{A,X\}.
 \label{eq:qSWAXTG}
\end{equation}
Here $A$ is one of the universal partial Casimirs.  Applying Lemma~\ref{lem:qSWtwoblock} to the elementary blocks $1$ and $2$, with $a=b_1$ and $b=b_2$, gives
\begin{equation}
 \{A,T\}=0,\qquad
 \{X,T\}=0,\qquad
 \{G,T\}=0.
 \label{eq:qSWTcentral}
\end{equation}
Accordingly, $T$ is central in the subalgebra generated by $A,X,G$.  At $N=3$ it should not be counted as a sixth functionally independent integral: the maximally superintegrable family already has the maximal rank $2N-1=5$.  Its role in the quadratic block is therefore that of a conserved central parameter.  The remaining nontrivial brackets close quadratically:
\begin{align}
 \{A,G\}={}&
 -8zA^2-16AX+8AT
 +8z(b_1+b_2)A
 +8(b_1-b_2)T,
 \label{eq:qSWclosure1}\\
 \{X,G\}={}&
 8X^2+16zAX-8XT
 +4\omega^2A
 -8z(b_1+b_2)X
 -8zb_1T+8z^2b_1b_2
 -4\omega^2(b_1+b_2).
 \label{eq:qSWclosure2}
\end{align}
Consequently, the missing St\"ackel integral generates a finite quadratic Poisson algebra rather than an open sequence of new functions.

\begin{proposition}
For the three-dimensional St\"ackel q-SW system, the functions $A,X,G$ defined by \eqref{eq:qSWAXTG} generate a quadratic Poisson algebra over the central parameters $T,b_1,b_2$.  A Casimir relation for this algebra is
\begin{align}
 \mathcal K_{\rm SSW}={}&
 G^2
 +16zA^2X
 +16AX^2
 -16ATX
 -16z(b_1+b_2)AX
 -16(b_1-b_2)TX
 \notag\\
 &+4\omega^2A^2
 +\left[
 -16zb_1T+16z^2b_1b_2
 -8\omega^2(b_1+b_2)
 \right]A
 \label{eq:qSWCasimirLeft}
\end{align}
with the central value
\begin{equation}
 \mathcal K_{\rm SSW}
 =
 -16b_1T^2
 +16zb_1(b_2-b_1)T
 -4\omega^2(b_1-b_2)^2.
 \label{eq:qSWCasimirValue}
\end{equation}
\end{proposition}

\begin{proof}
This is the specialization of Lemma~\ref{lem:qSWtwoblock} to the elementary blocks $1$ and $2$, with $a=b_1$, $b=b_2$, $A=L_{12}$, $X=Z_2$ and $T=Z_3$.  Hence both the quadratic brackets and the central value of $\mathcal K_{\rm SSW}$ follow at the abstract Poisson-coalgebra level; no phase-space realization is required for the proof.
\end{proof}

This provides a concrete interpretation of the algebraic results obtained earlier in the paper.  The universal rank-one deformed Racah sector generated by
\begin{equation}
 L_{12},\qquad L_{23},\qquad
 \mathcal F_{12,23}=\{L_{12},L_{23}\}
\end{equation}
is present for every Hamiltonian built from the same total coalgebra generators.  For the special St\"ackel q-SW Hamiltonian, the additional integral $Z_2$ enlarges this universal symmetry.  The two structures overlap through $L_{12}$, while
\begin{equation}
 \{Z_2,L_{23}\}=0
 \label{eq:qSWZ2L23}
\end{equation}
because the two functions have disjoint canonical support.  Hence the maximally superintegrable model contains, simultaneously, the universal deformed Racah algebra of Section~4 and the closed quadratic block \eqref{eq:qSWAXTG}--\eqref{eq:qSWclosure2} generated by the additional St\"ackel integral.  The complete closure of the union of these two overlapping subalgebras may require further cross-bracket generators, but maximal superintegrability and the existence of a finite closed algebra involving the missing integral are already explicit at this stage.

\subsubsection[Four-dimensional St\"ackel system]{The four-dimensional St\"ackel system and the complete coupling of $Z_2$ to the coalgebra family}
\label{sec:qsw4d}

For $N=4$ the St\"ackel family contains $Z_2,Z_3,Z_4$.  In addition to Eqs.~\eqref{eq:qSWZ2}--\eqref{eq:qSWZ3}, the third member is
\begin{align}
 Z_4={}&
 \frac{\sinh(zq_1^2)}{2zq_1^2}e^{zq_1^2+2z(q_2^2+q_3^2)}
 \left[p_1^2+b_1\left(\frac{zq_1}{\sinh(zq_1^2)}\right)^2\right]
 +\frac{\sinh(zq_2^2)}{2zq_2^2}e^{zq_2^2+2zq_3^2}
 \left[p_2^2+b_2\left(\frac{zq_2}{\sinh(zq_2^2)}\right)^2\right]
 \notag\\
 &+\frac{\sinh(zq_3^2)}{2zq_3^2}e^{zq_3^2}
 \left[p_3^2+b_3\left(\frac{zq_3}{\sinh(zq_3^2)}\right)^2\right]
 +\frac{\omega^2}{4z}\left(e^{2z(q_1^2+q_2^2+q_3^2)}-1\right).
 \label{eq:qSWZ4}
\end{align}
It is useful to extend the notation by
\begin{equation}
 Z_5:=\frac12 H^{\rm SSW}_{z,4}.
 \label{eq:qSWZ5}
\end{equation}
Thus $Z_{m+1}$ is one half of the St\"ackel Hamiltonian on the first $m$ sites, for $m=1,2,3,4$.  The quantities $Z_2,Z_3,Z_4,Z_5$ mutually Poisson commute.

Let
\begin{equation}
 X:=Z_2,
 \qquad
 A_2:=L_{12},\quad A_3:=L_{123},\quad A_4:=L_{1234},
 \label{eq:qSW4A}
\end{equation}
and introduce the complementary Casimirs
\begin{equation}
 B_2:=b_2,\qquad B_3:=L_{23},\qquad B_4:=L_{234},
 \label{eq:qSW4B}
\end{equation}
where $L_{123},L_{234},L_{1234}$ are the realized interval Casimirs listed in Appendix~A.  Finally set
\begin{equation}
 T_2:=Z_3,
 \qquad T_3:=Z_4,
 \qquad T_4:=Z_5,
 \qquad
 G_m:=\{A_m,X\},\quad m=2,3,4.
 \label{eq:qSW4GT}
\end{equation}
The five standard coalgebra integrals at $N=4$ may be taken as
\begin{equation}
 L_{12},\qquad L_{123},\qquad L_{1234},\qquad L_{34},\qquad L_{234}.
 \label{eq:qSW4coalgebrafamily}
\end{equation}
The direct coupling of the missing St\"ackel integral $X=Z_2$ to this complete family is
\begin{align}
 \{X,L_{12}\}&=-G_2,\qquad
 \{X,L_{123}\}=-G_3,\qquad
 \{X,L_{1234}\}=-G_4,
 \label{eq:qSW4directleft}\\
 \{X,L_{34}\}&=0,\qquad
 \{X,L_{234}\}=0. \label{eq:qSW4directright}
\end{align}
In addition, $\{X,L_{23}\}=0$.  The two zero brackets in \eqref{eq:qSW4directright} follow from disjoint canonical support.  Notice in particular that the total coalgebra Casimir $L_{1234}$, although central in the universal partial-Casimir algebra of Section~5, ceases to be central after the additional St\"ackel symmetry is adjoined.

The three nonzero brackets in \eqref{eq:qSW4directleft} do not generate an uncontrolled sequence.  For each $m=2,3,4$, the triple $A_m,X,G_m$ closes quadratically over the central elements $b_1,B_m,T_m$:
\begin{align}
 \{A_m,G_m\}={}&
 -8zA_m^2-16A_mX+8A_mT_m
 +8z(b_1+B_m)A_m
 +8(b_1-B_m)T_m,
 \label{eq:qSW4nested1}\\
 \{X,G_m\}={}&
 8X^2+16zA_mX-8XT_m+4\omega^2A_m
 -8z(b_1+B_m)X-8zb_1T_m
 \notag\\
 &+8z^2b_1B_m-4\omega^2(b_1+B_m).
 \label{eq:qSW4nested2}
\end{align}
Moreover,
\begin{equation}
 \{A_m,B_m\}=\{X,B_m\}=\{G_m,B_m\}=0,
 \qquad
 \{A_m,T_m\}=\{X,T_m\}=\{G_m,T_m\}=0.
 \label{eq:qSW4nestedcentral}
\end{equation}
For $m=2$, Eqs.~\eqref{eq:qSW4nested1}--\eqref{eq:qSW4nested2} are exactly the three-dimensional block \eqref{eq:qSWclosure1}--\eqref{eq:qSWclosure2}.  The genuinely new four-dimensional information is the pair of further blocks
\begin{equation}
 (A_3,B_3,T_3)=(L_{123},L_{23},Z_4),
 \qquad
 (A_4,B_4,T_4)=\left(L_{1234},L_{234},\frac12H^{\rm SSW}_{z,4}\right).
 \label{eq:qSW4newblocks}
\end{equation}
Thus the additional integral $Z_2$ couples nontrivially to every left coalgebra Casimir, while the complementary interval Casimir $B_m=C_{[2,m]}$ becomes a central parameter of the corresponding quadratic block.

Each block possesses a Casimir relation.  Define
\begin{align}
 \mathcal K_m={}&
 G_m^2+16zA_m^2X+16A_mX^2-16A_mT_mX
 -16z(b_1+B_m)A_mX-16(b_1-B_m)T_mX
 \notag\\
 &+4\omega^2A_m^2
 +\left[-16zb_1T_m+16z^2b_1B_m-8\omega^2(b_1+B_m)\right]A_m.
 \label{eq:qSW4nestedCasimir}
\end{align}
Then
\begin{equation}
 \mathcal K_m=
 -16b_1T_m^2+16zb_1(B_m-b_1)T_m
 -4\omega^2(b_1-B_m)^2,
 \qquad m=2,3,4.
 \label{eq:qSW4nestedCasimirValue}
\end{equation}
For $m=2$ this reduces to Eqs.~\eqref{eq:qSWCasimirLeft}--\eqref{eq:qSWCasimirValue}.  Equations~\eqref{eq:qSW4nested1}--\eqref{eq:qSW4nestedCasimirValue} therefore provide two additional nontrivial Casimir-bearing quadratic blocks at $N=4$.

\begin{proposition}\label{prop:qSW4nested}
For the four-dimensional St\"ackel q-SW system, the missing integral $Z_2$ has a closed quadratic coupling to every member of the standard left/right coalgebra family.  The only nonzero direct brackets with the five coalgebra integrals \eqref{eq:qSW4coalgebrafamily} are those in \eqref{eq:qSW4directleft}; each is governed by the quadratic relations \eqref{eq:qSW4nested1}--\eqref{eq:qSW4nested2} and the Casimir relation \eqref{eq:qSW4nestedCasimir}--\eqref{eq:qSW4nestedCasimirValue}.
\end{proposition}

\begin{proof}
Fix $m\in\{2,3,4\}$.  Regard site $1$ as the first Poisson-coalgebra block and the consecutive sites $2,\ldots,m$ as a second composite block.  By coassociativity and the Poisson-homomorphism property of the iterated coproduct, the total generators of the second block satisfy the same non-standard $\mathfrak{sl}_z(2)$ Poisson relations as one elementary copy, with Casimir $B_m=C_{[2,m]}$.  The Casimir of the combined two-block system is $A_m=C_{[1,m]}$.  Furthermore $X=Z_2=h_z(J^{(1)})$, while $T_m=Z_{m+1}=h_z(J^{([1,m])})$, with the convention \eqref{eq:qSWZ5} at $m=4$.  Lemma~\ref{lem:qSWtwoblock} therefore applies directly after
\begin{equation}
 A\mapsto A_m,\qquad a\mapsto b_1,\qquad b\mapsto B_m,\qquad T\mapsto T_m.
\end{equation}
This gives Eqs.~\eqref{eq:qSW4nested1}--\eqref{eq:qSW4nestedCasimirValue} without passing through a symplectic realization.  The right integrals $L_{34}$ and $L_{234}$ have support entirely on sites $2,3,4$ and hence Poisson-commute with $Z_2$, which depends only on site $1$.
\end{proof}

The same argument is not special to four sites.
\begin{corollary}[Arbitrary-$N$ nested coupling of the missing integral]\label{cor:qSWgeneralN}
For the $N$-site St\"ackel q-SW system, set
\begin{align}
 X&=Z_2, & A_m&=C_{[1,m]}, & B_m&=C_{[2,m]},\notag\\
 T_m&=Z_{m+1}, & G_m&=\{A_m,X\}, & m&=2,\ldots,N.
 \label{eq:qSWgeneralNdefs}
\end{align}
with $Z_{N+1}:=H^{\rm SSW}_{z,N}/2$.  Then for every $m$ the relations \eqref{eq:qSW4nested1}--\eqref{eq:qSW4nestedCasimirValue} hold with the same substitutions, and $X$ Poisson-commutes with every standard right coalgebra integral supported entirely on sites $2,\ldots,N$.  Hence these formulas give all direct brackets of the missing integral $Z_2$ with the standard $2N-3$ coalgebra family.
\end{corollary}
\begin{proof}
Apply Lemma~\ref{lem:qSWtwoblock} to the consecutive blocks $\{1\}$ and $[2,m]$ for each $m=2,\ldots,N$.  Their Casimirs are $b_1$ and $B_m$, their combined Casimir is $A_m$, and the one-block and combined St\"ackel functions are $Z_2$ and $Z_{m+1}$.  Coassociativity makes the composite block $[2,m]$ another exact copy of the same Poisson coalgebra.  The zero brackets with the right chain follow from disjoint site support.
\end{proof}

There is a second, overlapping chain that uses the whole St\"ackel family rather than keeping $Z_2$ fixed.  Set $A_k=C_{[1,k]}$, $X_k=Z_k$, $T_k=Z_{k+1}$ and $G_k^{\rm S}=\{A_k,X_k\}$, with $Z_5=H^{\rm SSW}_{z,4}/2$.  For $k=2,3,4$ the same equations \eqref{eq:qSW4nested1}--\eqref{eq:qSW4nested2} hold after
\begin{equation}
 (X,T_m,b_1,B_m)
 \longmapsto
 (X_k,T_k,L_{1\cdots k-1},b_k).
 \label{eq:qSW4successivesub}
\end{equation}
Thus the four-dimensional model contains three overlapping St\"ackel--coalgebra quadratic blocks, in addition to the universal four-site algebra of Section~5.  This nested structure is the finite-$N$ shadow of the general $N$-site construction discussed in the companion technical note.

These examples clarify the role of the coalgebra construction.  The $2N-3$ universal coalgebra integrals, together with the Hamiltonian, provide $2N-2$ constants that are not expected to exhaust the symmetry of every special Hamiltonian; rather, they form a common algebraic core.  When an additional hidden integral exists, as for the St\"ackel q-SW system, it can generate further polynomial substructures and complete the system to maximal superintegrability.  In this sense the deformed Racah algebra obtained in Sections~4 and 5 should be regarded as a universal subalgebra inside larger symmetry algebras associated with distinguished Hamiltonians.

\section{Conclusion}
We have studied superintegrable Hamiltonian systems and their symmetry algebras from the $\mathfrak{sl}(2)$ Lie-Poisson coalgebra and its non-standard Poisson deformation.  The main algebraic result is the finite quadratic closure, independent of any chosen realization, of the partial Casimir sector at the first two nontrivial ranks.    The universal closure calculation is performed directly on tensor products of the abstract Poisson algebra, rather than being inferred from one particular symplectic realization.  The resulting closure relations are universal identities for the chosen Poisson coalgebra.  

In the undeformed case, the three-site construction reproduces the rank-one Racah Poisson algebra, and the four-site calculation recovers the corresponding higher-rank structure.  This provides a useful reference for the deformed calculation and makes explicit how non-contiguous subset Casimirs simplify the primitive-coproduct presentation.  For the non-standard deformation the situation changes: the coproduct contains exponential factors and the same minimal-site basis is no longer available in a direct way.  At three sites the partial Casimirs and one bracket generator form a quadratic Poisson algebra, and at four sites the five connected partial Casimirs together with four bracket generators again close quadratically.

The finite quadratic closure is not a consequence of polynomial dependence of the elementary generators; the elementary expressions are transcendental in the Poisson variables.  It is instead a structural property of the partial Casimirs themselves, and every Poisson realization inherits the same algebra.  In the $N=3$ and $N=4$ calculations an additional simplification occurs: the final polynomial closure coefficients have no explicit $z$ dependence.  The deformation therefore changes the realization of the generators while leaving these displayed structure coefficients parameter-independent; at $N=3$ the fixed normalization~\cref{eq:nonstandard-normalization-limit,eq:nonstandard-CF-limit} reduces the relations exactly to the undeformed Racah form.  The canonical realization reviewed in Section~6 gives explicit Hamiltonians on spaces of non-constant curvature and provides concrete representatives of all generators.  Appendices~A and B make the realized formulae sufficiently explicit for independent symbolic checking.  The St\"ackel q-deformed Smorodinsky-Winternitz example makes the distinction between universal and additional symmetry explicit.  The coalgebra sector gives the quasi-maximal family common to the whole class, whereas the additional St\"ackel integral completes the system to maximal superintegrability.  Together with the partial Casimir $L_{12}$, this additional integral generates the closed quadratic algebra \,\eqref{eq:qSWclosure1}--\eqref{eq:qSWclosure2}, with Casimir relation \,\eqref{eq:qSWCasimirLeft}--\eqref{eq:qSWCasimirValue}.  At $N=4$ the same missing integral $Z_2$ has a complete quadratic coupling to the five standard coalgebra integrals: it Poisson-commutes with the two right integrals and produces three nested Casimir-bearing quadratic blocks with $L_{12}$, $L_{123}$ and $L_{1234}$.  The universal deformed Racah algebra therefore appears as one part of a larger, hierarchically organized symmetry structure specific to the St\"ackel model.

Several questions remain open.  The first is the extension of the universal closure calculation to arbitrary numbers of sites and the identification of a minimal but structurally transparent generating set at higher rank.  A second question is the complete closure of the union of the full connected-interval deformed Racah sector, including its bracket generators, with the St\"ackel sector for the maximally superintegrable q-Smorodinsky-Winternitz model.  The direct coupling of the missing integral to the standard coalgebra family is now explicit at $N=4$ and extends to arbitrary $N$ by Lemma~\ref{lem:qSWtwoblock} together with coassociativity, but the closure of the larger union containing every universal bracket generator remains a separate problem.  It would also be useful to determine the Casimir relations and symplectic leaves of these enlarged algebras and to compare them with the undeformed $SW(N)$ structures.  Comparison with standard Drinfeld--Jimbo Poisson and quantum deformations may clarify how the present non-standard deformation is related to braided intermediate Casimirs and Askey--Wilson-type structures.

\section*{Acknowledgment}

\noindent IM was supported by the Australian Research Council Future Fellowship FT180100099. 

\section*{Data availability}

\noindent The manuscript has no associated data.

\appendix
\section[Explicit higher-site partial Casimirs for N=4]{Explicit higher-site partial Casimirs for $N=4$}
\label{app:higher-casimirs}

    The pair kernel $I^z_{ij}$ in Eq.~\eqref{eq:pairkernel} gives the higher-site partial Casimirs in a form that is both explicit and readable:
\begin{align}
 L_{123}={}&
 b_1e^{2z(q_2^2+q_3^2)}
 +b_2e^{-2z(q_1^2-q_3^2)}
 +b_3e^{-2z(q_1^2+q_2^2)}\notag\\
 &+I^z_{12}e^{-z(q_1^2-q_2^2-2q_3^2)}
 +I^z_{13}e^{-z(q_1^2-q_3^2)}
 +I^z_{23}e^{-z(2q_1^2+q_2^2-q_3^2)},
 \label{eq:L123compact}
\end{align}
\begin{align}
 L_{234}={}&
 b_2e^{2z(q_3^2+q_4^2)}
 +b_3e^{-2z(q_2^2-q_4^2)}
 +b_4e^{-2z(q_2^2+q_3^2)}\notag\\
 &+I^z_{23}e^{-z(q_2^2-q_3^2-2q_4^2)}
 +I^z_{24}e^{-z(q_2^2-q_4^2)}
 +I^z_{34}e^{-z(2q_2^2+q_3^2-q_4^2)},
 \label{eq:L234compact}
\end{align}
\begin{align}
 L_{1234}={}&
 b_1e^{2z(q_2^2+q_3^2+q_4^2)}
 +b_2e^{-2z(q_1^2-q_3^2-q_4^2)}+b_3e^{-2z(q_1^2+q_2^2-q_4^2)}
 +b_4e^{-2z(q_1^2+q_2^2+q_3^2)}\notag\\
 &+I^z_{12}e^{-z(q_1^2-q_2^2-2q_3^2-2q_4^2)}
 +I^z_{13}e^{-z(q_1^2-q_3^2-2q_4^2)}+I^z_{14}e^{-z(q_1^2-q_4^2)}\notag\\
 &+I^z_{23}e^{-z(2q_1^2+q_2^2-q_3^2-2q_4^2)}
 +I^z_{24}e^{-z(2q_1^2+q_2^2-q_4^2)}
 +I^z_{34}e^{-z(2q_1^2+2q_2^2+q_3^2-q_4^2)}.
 \label{eq:L1234compact}
\end{align}
Equations~\eqref{eq:L123compact}--\eqref{eq:L1234compact} are exactly equivalent to the fully expanded expressions, but the repeated phase-space dependence is isolated once and for all in $I^z_{ij}$. This form is better suited to verification, comparison and subsequent differentiation.

\subsection{Exact phase-space form of the four bracket generators}
The four bracket generators are most compactly written directly as canonical Poisson brackets of the realized partial Casimirs.  This form is exact, avoids unnecessary repeated exponential factors, and can be expanded algebraically when required.

With the convention~\eqref{canonicalPBcorrected} and with the explicit functions $L_A=D_z(C_A)$ above,
\begin{align}
 \cF_{12,23}&=\{L_{12},L_{23}\}
 =\frac{\partial L_{12}}{\partial q_2}\frac{\partial L_{23}}{\partial p_2}
 -\frac{\partial L_{12}}{\partial p_2}\frac{\partial L_{23}}{\partial q_2},\\
 \cF_{23,34}&=\{L_{23},L_{34}\}
 =\frac{\partial L_{23}}{\partial q_3}\frac{\partial L_{34}}{\partial p_3}
 -\frac{\partial L_{23}}{\partial p_3}\frac{\partial L_{34}}{\partial q_3},\\
 \cF_{12,234}&=\{L_{12},L_{234}\}
 =\frac{\partial L_{12}}{\partial q_2}\frac{\partial L_{234}}{\partial p_2}
 -\frac{\partial L_{12}}{\partial p_2}\frac{\partial L_{234}}{\partial q_2},\\
 \cF_{34,123}&=\{L_{34},L_{123}\}
 =\frac{\partial L_{34}}{\partial q_3}\frac{\partial L_{123}}{\partial p_3}
 -\frac{\partial L_{34}}{\partial p_3}\frac{\partial L_{123}}{\partial q_3}.
 \label{eq:exactrealizedF}
\end{align}
Only the common site contributes in each bracket, which explains the one-term derivative formulae.  These expressions are exact, avoid a fragile multi-page expansion, and reproduce the abstract closure relations after the substitution $C_i=b_i$ and $C_A=L_A$.

\section{Factorized phase-space form of the bracket generators}
\label{app:factorized-brackets}

The compact expressions in Appendix~A are preferable for the main calculations, but a phase-space formula in which every derivative has been evaluated is also useful.  This appendix gives the four bracket generators in factorized form.  The result is algebraically equivalent to a fully expanded expression, while making the common structure of the four generators visible and reducing the risk of transcription errors.

Introduce the abbreviations
\begin{equation}
 S_i:=\sinh(zq_i^2),\qquad R_i:=\cosh(zq_i^2),\qquad
 A_i:=\frac{S_i}{zq_i^2},\qquad
 \Lambda_{ij}:=q_ip_j-q_jp_i,
\end{equation}
and
\begin{equation}
 E_{ij}:=e^{-z(q_i^2-q_j^2)}.
\end{equation}
Then the pair kernel can be written as
\begin{equation}
 I^z_{ij}=A_iA_j\Lambda_{ij}^2+b_i\frac{S_j}{S_i}+b_j\frac{S_i}{S_j}.
\end{equation}
The elementary derivatives required below are
\begin{align}
 \partial_{p_i}I^z_{ij}&=-2q_jA_iA_j\Lambda_{ij},
 \qquad \partial_{p_j}I^z_{ij}=2q_iA_iA_j\Lambda_{ij},\label{eq:Ipderivatives}\\
 \partial_{q_i}I^z_{ij}&=
 A_j\left(\frac{2R_i}{q_i}-\frac{2S_i}{zq_i^3}\right)\Lambda_{ij}^2
 +2A_iA_jp_j\Lambda_{ij}
 -2zb_iq_i\frac{S_jR_i}{S_i^2}
 +2zb_jq_i\frac{R_i}{S_j},\label{eq:Iqiderivative}\\
 \partial_{q_j}I^z_{ij}&=
 A_i\left(\frac{2R_j}{q_j}-\frac{2S_j}{zq_j^3}\right)\Lambda_{ij}^2
 -2A_iA_jp_i\Lambda_{ij}
 +2zb_iq_j\frac{R_j}{S_i}
 -2zb_jq_j\frac{S_iR_j}{S_j^2}.\label{eq:Iqjderivative}
\end{align}
These formulae are obtained by direct differentiation of Eq.~\eqref{eq:pairkernel}.

For a two-site partial Casimir
\begin{equation}
 L_{ij}=b_i e^{2zq_j^2}+b_j e^{-2zq_i^2}+I^z_{ij}E_{ij},
\end{equation}
the derivatives with respect to its two sites are
\begin{align}
 P^{(i)}_{ij}&:=\partial_{p_i}L_{ij}=E_{ij}\partial_{p_i}I^z_{ij},\qquad
 P^{(j)}_{ij}:=\partial_{p_j}L_{ij}=E_{ij}\partial_{p_j}I^z_{ij},\label{eq:Pijexplicit}\\
 Q^{(i)}_{ij}&:=\partial_{q_i}L_{ij}
 =-4zq_i b_j e^{-2zq_i^2}
 +E_{ij}\left(\partial_{q_i}I^z_{ij}-2zq_iI^z_{ij}\right),\label{eq:Qijiexplicit}\\
 Q^{(j)}_{ij}&:=\partial_{q_j}L_{ij}
 =4zq_j b_i e^{2zq_j^2}
 +E_{ij}\left(\partial_{q_j}I^z_{ij}+2zq_jI^z_{ij}\right).\label{eq:Qijjexplicit}
\end{align}
Substitution of Eqs.~\eqref{eq:Ipderivatives}--\eqref{eq:Iqjderivative} makes the right-hand sides functions only of $(q_i,p_i,q_j,p_j)$, $z$ and $b_i,b_j$.

The first two bracket generators are
\begin{align}
 \cF_{12,23}&=Q^{(2)}_{12}P^{(2)}_{23}-P^{(2)}_{12}Q^{(2)}_{23},\label{eq:F1223factorized}\\
 \cF_{23,34}&=Q^{(3)}_{23}P^{(3)}_{34}-P^{(3)}_{23}Q^{(3)}_{34}.\label{eq:F2334factorized}
\end{align}
No derivatives remain in these expressions because all $P$ and $Q$ blocks are given explicitly above.

For the three-site quantities define
\begin{align}
 E^{(234)}_{23}&:=e^{-z(q_2^2-q_3^2-2q_4^2)},&
 E^{(234)}_{24}&:=e^{-z(q_2^2-q_4^2)},&
 E^{(234)}_{34}&:=e^{-z(2q_2^2+q_3^2-q_4^2)},\\
 E^{(123)}_{12}&:=e^{-z(q_1^2-q_2^2-2q_3^2)},&
 E^{(123)}_{13}&:=e^{-z(q_1^2-q_3^2)},&
 E^{(123)}_{23}&:=e^{-z(2q_1^2+q_2^2-q_3^2)}.
\end{align}
Using Appendix~A one obtains the explicit derivatives
\begin{align}
 P^{(2)}_{234}:=\partial_{p_2}L_{234}
 &=E^{(234)}_{23}\partial_{p_2}I^z_{23}
 +E^{(234)}_{24}\partial_{p_2}I^z_{24},\label{eq:P234explicit}\\
 Q^{(2)}_{234}:=\partial_{q_2}L_{234}
 &=-4zq_2b_3e^{-2z(q_2^2-q_4^2)}
-4zq_2b_4e^{-2z(q_2^2+q_3^2)}+E^{(234)}_{23}\left(\partial_{q_2}I^z_{23}-2zq_2I^z_{23}\right)\nonumber\\
 &\quad
+E^{(234)}_{24}\left(\partial_{q_2}I^z_{24}-2zq_2I^z_{24}\right)-4zq_2E^{(234)}_{34}I^z_{34},\label{eq:Q234explicit}\\
 P^{(3)}_{123}:=\partial_{p_3}L_{123}
 &=E^{(123)}_{13}\partial_{p_3}I^z_{13}
 +E^{(123)}_{23}\partial_{p_3}I^z_{23},\label{eq:P123explicit}\\
 Q^{(3)}_{123}:=\partial_{q_3}L_{123}
 &=4zq_3b_1e^{2z(q_2^2+q_3^2)}
   +4zq_3b_2e^{-2z(q_1^2-q_3^2)}+4zq_3E^{(123)}_{12}I^z_{12}\nonumber\\
 &\quad +E^{(123)}_{13}\left(\partial_{q_3}I^z_{13}+2zq_3I^z_{13}\right)+E^{(123)}_{23}\left(\partial_{q_3}I^z_{23}+2zq_3I^z_{23}\right).\label{eq:Q123explicit}
\end{align}
Again, the derivatives of $I^z_{ij}$ appearing on the right are the explicit expressions in Eqs.~\eqref{eq:Ipderivatives}--\eqref{eq:Iqjderivative}.  The remaining two bracket generators are
\begin{align}
 \cF_{12,234}&=Q^{(2)}_{12}P^{(2)}_{234}-P^{(2)}_{12}Q^{(2)}_{234},\label{eq:F12234factorized}\\
 \cF_{34,123}&=Q^{(3)}_{34}P^{(3)}_{123}-P^{(3)}_{34}Q^{(3)}_{123}.\label{eq:F34123factorized}
\end{align}
Equations~\eqref{eq:F1223factorized}, \eqref{eq:F2334factorized}, \eqref{eq:F12234factorized} and \eqref{eq:F34123factorized} are the fully explicit factorized phase-space expressions used for the realized four-site algebra.  They are preferable to a several-page fully expanded transcription because the latter obscures the common kernel structure.

\end{document}